\documentclass[aps,reprint,amsfonts, amssymb, amsmath,pra, superscriptaddress, twocolumn,longbibliography,nofootinbib]{revtex4-1}

\usepackage{xcolor}

\usepackage{float}

\usepackage[shortlabels]{enumitem}

\usepackage{braket}
\usepackage{amsthm}
\usepackage{mathtools}
\usepackage{url}
\usepackage{physics}
\usepackage{graphicx}
\usepackage[left=16mm,right=16mm,top=35mm,columnsep=15pt]{geometry} 

\usepackage[T1]{fontenc}

\usepackage{bm}

\usepackage{silence}
\newcommand*{\eh}{\mathrm{End\, }(\mathcal{H})}

\def\ad{^{\dagger}}

\newcommand{\fsnull}[1]{}
\newcommand{\old}[1]{}

\addtocontents{toc}{\protect\setcounter{tocdepth}{1}}

\usepackage{hyperref}
\usepackage[toc,page,header]{appendix}

\addtocontents{toc}{\protect\setcounter{tocdepth}{1}}

\usepackage[makeroom]{cancel}
\definecolor{C1}{RGB}{52, 89, 149}
\definecolor{C2}{RGB}{251, 77, 61}
\definecolor{C3}{RGB}{3, 206, 164}
\definecolor{C4}{RGB}{202, 21, 81}
\definecolor{C5}{RGB}{202, 21, 81}
\hypersetup{colorlinks=true, linkcolor=C5, citecolor=C5, urlcolor=C5}

\usepackage{tikz}
\tikzset{every picture/.style=remember picture}

\usepackage[utf8]{inputenc}
\usepackage{graphicx}
\usepackage{xcolor}
\usepackage{amsmath}
\usepackage{amsthm}
\usepackage{bm}
\usepackage{bbm}
\usepackage{comment}
\usepackage{mathdots}
\usepackage{lipsum}
\usepackage{verbatim}
\usepackage{natbib}
\usepackage{nccmath}
\usepackage{amsfonts}
\usepackage{thm-restate}
\usepackage{thmtools}
\usepackage{ytableau}

\usepackage{amssymb}
\usepackage{dsfont}

\newcommand{\Var}{{\rm Var}}

\newcommand{\ot}{\otimes}
\newcommand{\ts}{^{\otimes 2}}

\newcommand{\bs}{\textsf{BS}}

\newcommand{\lm}{\lambda }

\newcommand{\Lm}{\Lambda }

\def\calL{\mathcal{L}}

\def\C{\mathbb{C}}
\newcommand{\mcl}{\mathcal{L}}

\newcommand{\mcw}{\mathcal{W}}

\newcommand{\mco}{\mathcal{O}}

\newcommand{\mch}{\mathcal{H}}
\newcommand{\mcm}{\mathcal{M}}

\newcommand{\mce}{\mathcal{E}}

\newcommand{\mcz}{\mathcal{Z}}

\newcommand{\mbc}{\mathbb{C}}
\newcommand{\mbr}{\mathbb{R}}

\newcommand{\mbe}{\mathbb{E}}
\newcommand{\mst}{\mathsf{T}}

\newcommand{\mfu}{\mathfrak{u}}

\def\be{\begin{equation}}
\def\ee{\end{equation}}
\def\bs{\begin{split}}
\def\e{\end{split}}
\def\ba{\begin{eqnarray}}
\def\bea{\begin{eqnarray}}

\def\tea{\end{eqnarray}}
\def\ea{\end{eqnarray}}
\def\eea{\end{eqnarray}}

\def\d{\delta}

\def\d{\delta}

\def\tk{^\otimes k}

\def\tt{^{\otimes 3}}

\def\tk{^{\otimes k}}

\newcommand{\id}{\mathds{1}}

\renewcommand{\d}{\delta}

\addtocontents{toc}{\protect\setcounter{tocdepth}{0}}

\DeclareMathOperator*{\expect}{\mathbb{E}}

\def\be{\begin{equation}}
\def\te{\end{equation}}
\def\ee{\end{equation}}
\def\ba{\begin{eqnarray}}
\def\bea{\begin{eqnarray}}

\def\tea{\end{eqnarray}}
\def\ea{\end{eqnarray}}
\def\eea{\end{eqnarray}}

\begin{document}

\makeatletter
\newif\ifinappendixtoc  

\newcommand{\tableofcontentsappendixonly}{%
  \begingroup
    \inappendixtocfalse

    \@ifundefined{l@section}{}{%
      \let\ao@l@section\l@section
      \def\l@section##1##2{\ifinappendixtoc \ao@l@section{##1}{##2}\fi}%
    }
    \@ifundefined{l@subsection}{}{%
      \let\ao@l@subsection\l@subsection
      \def\l@subsection##1##2{\ifinappendixtoc \ao@l@subsection{##1}{##2}\fi}%
    }
    \@ifundefined{l@subsubsection}{}{%
      \let\ao@l@subsubsection\l@subsubsection
      \def\l@subsubsection##1##2{\ifinappendixtoc \ao@l@subsubsection{##1}{##2}\fi}%
    }

    \def\AppendixTOCMark{\global\inappendixtoctrue}%

    \tableofcontents
  \endgroup
}
\makeatother

\title{
Particle-preserving fermionic shadows with mode-independent sample complexity
}

\author{Maxwell West}
\thanks{westm@lanl.gov}
\affiliation{Theoretical Division, Los Alamos National Laboratory, Los Alamos, New Mexico 87545, USA}
\affiliation{School of Physics, University of Melbourne, Parkville, VIC 3010, Australia}

\author{M. Cerezo}
\affiliation{Information Sciences, Los Alamos National Laboratory, Los Alamos, New Mexico 87545, USA}
\affiliation{Quantum Science Center, Oak Ridge, TN 37931, USA}

\author{Mart\'{i}n Larocca}
\affiliation{Theoretical Division, Los Alamos National Laboratory, Los Alamos, New Mexico 87545, USA}
\affiliation{Quantum Science Center, Oak Ridge, TN 37931, USA}

\begin{abstract}
We consider the problem of learning expectation values of particle-preserving operators with respect to  an unknown $\eta$-particle $n$-mode fermionic state  via classical shadows. Our main application is to estimating overlaps with arbitrary Slater determinant states: While it is known that such overlaps can, in the average case, be learnt to a fixed additive precision with a constant number of samples, the  best-known worst case bound is $\mco(\sqrt n \log n)$; here we improve this to  $\mco(\eta\log\eta)$, achieving a mode-independent sample cost. Our procedure is also computationally efficient, requiring only classical post-processing which for a generic dense orbital runs in time $\mco(n\eta^2)$. For the task of estimating the expectation value of a general particle-preserving quadratic
fermionic observable $h$, we prove a sample complexity bound of $\mco(\eta \|h_0\|_2^2)$, where $h_0$ is the traceless component of $h$; 
the associated classical post-processing scales as $\mco(n^2\eta )$. 
Finally, we discuss implementation of the required randomization: in a first-quantized encoding, approximate unitary designs give circuit depths polylogarithmic in the number of modes, contrasting with linear-depth requirements for nearest-neighbor second-quantized matchgate implementations. On the technical side, our proof reduces the extremal shadow variance to harmonic analysis on the AIII symmetric space $U(n)/(U(\eta)\times U(n-\eta))$ and evaluates the resulting integral using techniques from the theories of Jacobi ensembles and orthogonal polynomials, in a calculation which may be of independent interest.

\end{abstract}

\maketitle

\section{Introduction}
Learning certain classes of expectation values with respect to an unknown quantum state, while consuming as few copies of that state as possible, is a fundamental task in quantum information theory~\cite{huang2020predicting,aaronson2018shadow}. For example, one is often interested in learning observables (including indeed states themselves) that are somehow of ``low rank'', either in the literal linear algebraic sense, or more generally~\cite{haah2017sample,scharnhorst2025optimal}. In this spirit, when studying fermionic systems an important problem is efficiently learning $k$-reduced density matrices~\cite{low2022classical,wan2023matchgate,zhao2021fermionic}, or, more specifically, the overlap of the unknown state with Slater determinant states.  

In recent years, \textit{classical shadows}~\cite{huang2020predicting,aaronson2018shadow,west2026classical,bertoni2024shallow,low2022classical,zhao2021fermionic,wan2023matchgate,van2022hardware,ippoliti2024classical,zhao2024group,king2024triply,jerbi2023shadows,koh2022classical,chen2021robust,hearth2024efficient,chan2022algorithmic,sauvage2024classical,Helsen_2023,kunjummen2023shadow,denzler2023learningfermioniccorrelationsevolving,west2024random,grier2024sample,brandao2020fast,bertoni2022shallow,vitale2024estimation,somma2024shadow,west2025real,bringewatt2025classical} has emerged as a leading framework within which to perform such learning, with various instantiations of its general theory now understood to perform well for various types of target observables~\cite{west2026classical}. A classical shadow protocol involves measuring copies of the unknown state in randomized bases, with the specific randomization employed being protocol-dependent, and informed by which observables one wishes to learn. In the aforementioned fermionic case, both measurements based on evolving with random matchgates~\cite{wan2023matchgate,zhao2021fermionic} and particle-preserving fermionic unitaries~\cite{low2022classical} have been employed. 

The first main result of this work is to extend the analysis of Ref.~\cite{low2022classical} to show that particle-preserving fermionic shadows can learn (with a fixed high probability and to a fixed additive   precision, qualifiers we henceforth take as read) the overlap of an unknown $\eta$-particle $n$-mode fermionic state $\rho$ with an arbitrary Slater determinant state using $\mco(\eta\log\eta)$ copies of $\rho$. This can be thought of as a worst-case analogue of the result of Ref.~\cite{low2022classical}, which proves that   the corresponding average case problem can be solved with constant sample complexity (see Table~\ref{tab:tab1}). This bound is tight in the sense that we demonstrate the existence of states for which the sample complexity of this shadows framework indeed goes like $\eta\log\eta$ (up to a multiplicative constant which can be taken arbitrarily close to one as $n\to\infty$). This may also be compared with the (particle non-preserving) fermionic shadows  of Ref.~\cite{wan2023matchgate}, where the sample complexity of Slater overlap determination is bounded as $\mco(\sqrt{n}\log n)$. For $\eta\ll \sqrt n$, then, our bound is significantly tighter.  Our procedure also incurs only favorably scaling classical post processing costs, with the classical computation required to turn a single-shot shadow into an estimate of a Slater overlap scaling as $\mco(n\eta^2)$ for a generic dense orbital. This requires handling the various operators involved with some care, as they act on Hilbert spaces of (the generally much larger) dimension $\binom{n}{\eta}$. 

Secondly, we turn particle-preserving fermionic shadows to the task of  estimating arbitrary particle-preserving quadratic fermionic observables. For such an observable $h$, we show that the required number of samples  scales as $\mco(\eta \|h_0\|_2^2)$, where $h_0=h-(\Tr[h]/n)\id$ is the traceless component of $h$. Interestingly, while our approach to bounding the sample complexity in the Slater determinant case requires some reasonably technical results from the theory of zonal spherical functions on the so-called symmetric space of type AIII, our analysis for general particle-preserving quadratic fermionic observables requires only significantly more elementary techniques; we will see that this is a consequence of the different ways that the two cases interact with the representation theory of our setup.

Finally, we discuss the circuit depths required to implement the random unitaries involved in our shadow protocol. This  is the first moment at which it becomes important whether one is working in first or second quantization, with our previously quoted results agnostic to this detail. We employ recent results~\cite{schuster2024random,cui2025unitary,west2025no,grevink2025will} from the theory of random circuits to show that in first quantization one can achieve circuit depths of merely $\mco(\log\eta\log\log(\eta\log n))$, whereas in second quantization one has (under some natural to-be-detailed  assumptions) matching upper and lower bounds of $\Theta(n)$. 

\begin{table}[]
\begin{tabular}{|c|c|c|c|}
\hline
 Reference & Slater  & Quadratic $h$ & Case \\ \hline
 Ref.~\cite{low2022classical} & ${ \mco\big(\frac{n+1}{n-\eta+1}\big)}$ & -- & Average \\ \hline
 Ref.~\cite{wan2023matchgate} & $\mco(\sqrt{n}\log n)$  & $\mco(n\|h_0\|_2^2)$  & Worst   \\ \hline
 \ \ This work \ \ &\ \ $\mco(\eta\log\eta)$\ \ &\ \ $\mco(\eta\|h_0\|_2^2)$\ \ &  Worst   \\ \hline
\end{tabular}
\caption{Comparison of our sample complexity results with those of previously existing fermionic shadow schemes for our two families of observables, along with whether an average-case or worst-case bound is considered. 
The most direct comparison is with Ref.~\cite{wan2023matchgate}, over which we obtain clear improvements   in the case $\eta\ll \sqrt n$. 
}\label{tab:tab1}
\end{table}

\section{Preliminaries}

We begin by establishing our notation and reviewing some basic ideas from the theory of  classical shadows, in particular in the context of the example of the protocol given in Ref.~\cite{low2022classical}. We will be considering  $\eta$-particle states of an $n$-mode fermionic system,  the relevant Hilbert space for which is $\mch_\eta=\Lm^{\eta}(\mbc^n)$. We will always take $1\le\eta\le n/2$, although due to particle-hole symmetry our results carry through with only trivial modifications for $\eta>n/2$; in that regime our  bounds will be functions of $n-\eta$ instead of $\eta$. Given a basis $e_1,e_2,\ldots,e_n$ of $\mbc^n$,  a \textit{coordinate Slater determinant} $\ket Z$ is one of the occupation-basis vectors $\ket Z  = e_{z_1} \wedge \cdots \wedge  e_{z_\eta} $, where $Z=\{z_1<\cdots<z_\eta\}\subset[n]$; we denote  the corresponding density matrix by  $\Pi_Z=\ketbra Z$. We call the set of length-$\eta$ subsets of $[n]$ to which $Z$ belongs $\mcz_\eta$ (so that $|\mcz_\eta| = \binom{n}{\eta}$).  More generally, an $\eta$-particle Slater determinant $\ket\phi = U_\eta(v)\ket R $ may be obtained by rotating the one-particle orbitals; here $R=[\eta] $ is a fixed reference occupation and $U_\eta(v)=\Lambda^\eta v$. We will also make frequent use of the \textit{Johnson distance} between two $\eta$-subsets $S$ and $Z$, given by $d(S,Z)=\eta-|S\cap Z|$. 

We turn next to a brief review of the  points from classical shadows needed in this work; more details may be found in e.g. Refs.~\cite{huang2020predicting,west2026classical}. The setup is as follows. First, one picks (i) an ensemble $\mce$ of unitaries, acting on the Hilbert space $\mch$ of the unknown state via some representation $\pi: u\in\mce\mapsto \pi(u)\in U(\mch)$, the group of unitaries on $\mch$, and (ii) a measurement basis $\mcw$ of $\mch$. One then evolves copies of one's unknown state under the action of unitaries sampled from $\mce$, and measures the resulting state in the basis $\mcw$. Upon sampling a unitary $u\in\mce$ and measuring the outcome $\ket w\in\mcw$, one records the state $\pi(u) \ketbra{w}{w} \pi(u)\ad$. In this work we will, following Ref.~\cite{low2022classical}, exclusively make the choices $\pi=U_\eta$ and $\mcw=\mcz_\eta$. The upshot  of all this is to effect the \textit{measurement channel }
\begin{equation}\label{eq:chan}
    \mcm (\rho) = \hspace{-1mm} \expect_ {u\sim \mce } \sum_{Z\in\mcz_\eta} \Tr[\rho U_\eta(u)\ad \Pi_ZU_\eta(u)]U_\eta(u)\ad \Pi_Z U_\eta(u)\,. 
\end{equation}
While the (standard) description of the classical shadow channel that we have just given may  seem like something of an \textit{ad hoc} recipe, it is in fact describing a very natural map~\cite{west2026classical}. Indeed, let $\Xi(\rho) = \sum_{w\in\mcw} \Tr[\rho \Pi_w]\Pi_w$ be the completely dephasing channel with respect to the basis $\mcw$. Then Eq.~\eqref{eq:chan} is simply the (average over $\mce$) of the natural action of our  unitaries on $\Xi\in {\rm End}({\rm End}(\mch))$; that is, we have $\mcm = \expect_{u\sim \mce} u \cdot \Xi$.

Next, for a given observed pair $(u,Z)$, one defines the corresponding classical shadow of $\rho$ to be $\hat\rho(u,Z) = \mcm^{-1}(U_\eta(u)\ad \Pi_ZU_\eta(u))$, and obtains a corresponding estimator $\hat o (u,Z) = \Tr[\hat\rho (u,Z) O]$ of the expectation value $ o = \Tr[\rho O]$. In our case these estimators will always be unbiased; this follows from the fact, proved in Ref.~\cite{low2022classical}, that our measurement channel is tomographically complete. Given that taking the average of estimates coming from empirically sampled shadows will therefore converge to the true expectation values, a key remaining question is of the speed of convergence. This speed is controlled by the \textit{variance} of the estimators; indeed one can show that in order to estimate all of the expectation values $\{\Tr[\rho O_i]\}_{i=1}^K$ to within (with high probability) an additive precision $\varepsilon$   requires $\mco(\varepsilon^{-2}\log K\max_i \Var[\hat{o}_i])$ shadows~\cite{huang2020predicting}.  Calculating such variances is in this work our primary preoccupation.

\section{Results}
\noindent
We begin this section by sketching our proof of our main result:
\begin{restatable}{thm}{thmone}\label{thm:1} 
Let $\rho$ be an $\eta$-particle $n$-mode fermionic state, and let $\Pi_{\phi}$ be an $\eta$-particle Slater determinant state with $\eta\le n/2$. Then we have 
\begin{equation}
    \Var [\hat{\Pi}_{\phi}] \le (\eta + 1)H_{\eta+1}\,,
\end{equation}
where $H_m=\sum_{j=1}^m1/j\approx \log m$ is the $m$\textsuperscript{th} harmonic number. 
\end{restatable}
By the discussion of the previous section, this translates directly to a statement of the sample-efficiency with which one can estimate overlaps with Slater determinants.  
The full proof of Theorem~\ref{thm:1}, which is spread over  Appendices~\ref{sec:sd} and~\ref{sec:zsf}, turns out to be somewhat technical; at a high-level, however, the outline of the argument is quite straightforward: 

\begin{proof}[Proof of Theorem~\ref{thm:1} (sketch)] 
We begin by showing that we can without loss of generality take $\ket\phi$  to be the reference state $\ket R=\wedge_i e_i$; this follows from the transitivity of the action of $U_\eta$ on the set of Slater determinant states and the invariance of the Haar measure on the unitary group. 

Next, we use the analytical characterisation of the channel Eq.~\ref{eq:chan} given in Ref.~\cite{low2022classical} to evaluate the inverse channel on $\Pi_R$, finding
\begin{equation}\label{eq:minvpr}
    \mcm^{-1}(\Pi_R) =   \sum_{t=0}^\eta \beta_t
  \sum_{\substack{T\in\mcz_\eta\\ d(T,R)=t}}\Pi_T\,,
\end{equation}
where $\beta_t=(-1)^t\binom{n-t}{\eta-t} / \binom{\eta}{t}$. We then argue that the inverse channel taking this form implies that  the estimator corresponding to an outcome $(u,Z)$ can be written as a function
\begin{equation}
    F(v_Z\ad u) = \Tr\left[\Pi_RU_\eta(v_Z\ad u)^\dagger \mcm^{-1}(\Pi_R) U_\eta(v_Z\ad u)\right]
\end{equation}
of the ``relative unitary'' $v_Z\ad u$, where $v_Z$ is such that $v_ZR=Z$. 

Now, from $\Var [\widehat\Pi_R] = \expect [\widehat\Pi_R^2] - \expect[\widehat\Pi_R]^2\le \expect [\widehat\Pi_R^2] $, it is sufficient to bound the operator norm of 
\begin{equation}\label{eq:gammar}
      \Gamma_R=\int_{u\sim  U(n)}\sum_{Z\in\mcz_\eta}
\widehat\Pi_R(u,Z)^2U_\eta(u)^\dagger\Pi_ZU_\eta(u)\, ;
\end{equation}
indeed,  $\expect [\widehat\Pi_R^2]=\Tr[\Gamma_R \rho]\le \|\Gamma_R\|_\infty \|\rho\|_1= \|\Gamma_R\|_\infty $, where we have recalled H{\"o}lder's inequality. Our analysis of the spectrum of $\Gamma_R$ is facilitated by the observation that it commutes with the relevant action of $K= U(\eta)\times U(n-\eta)\subset  U(n)$, i.e. $  U_\eta(k)\Gamma_RU_\eta(k)^\dagger=\Gamma_R
$ for $k\in K$. Under this action of $K$, the $\eta$-particle Hilbert space decomposes into (pairwise inequivalent) irreps as $\mch_\eta=\bigoplus_{s=0}^\eta\mch_s$, where $\mch_s={\rm span}\,\{|S\rangle:d(S,R)=s\}$ is the \textit{Johnson shell} at distance $s$ from $R$, of dimension $N_s= \binom{\eta}{s}\binom{n-\eta}{s}$. Schur's lemma then implies that we have $  \Gamma_R|_{\mch_s}=\mu_s \id_{\mch_s}$ for some scalars $\mu_s$; the operator norm of $\Gamma_R$ is $  \|\Gamma_R\|_{\infty }=\max_{0\le s\le\eta}|\mu_s|$.

Our evaluation of the eigenvalues $\mu_s$  proceeds by expressing them as integrals over the symmetric space $G/K =  U(n)/( U(\eta)\times U(n-\eta))$ of type AIII. Indeed, having defined for $x=gK\in G/K$ the average  transition probability from $R$ to $S$ under $g$ as 
\begin{equation}
    p_s(gK)=\frac{1}{N_s}
  \sum_{\substack{S\in\mcz_\eta\\d(S,R)=s}}|\langle S|U_\eta(g)|R\rangle|^2
\end{equation}
(which one checks is well-defined as a function on $G/K$) we find that
\begin{equation}\label{eq:muint}
    \mu_s=\binom n\eta\int_{x\sim G/K}F(x)^2p_s(x)\,.
\end{equation}
The benefit of casting the eigenvalues in this form is that we can invoke the powerful theory of \textit{zonal spherical functions} on AIII, which we use to show that $\mu_s\le\mu_0$ for all $s$ (note it follows from Eq.~\eqref{eq:muint} that all $\mu_s\ge 0$). 

Finally, it remains to evaluate $\mu_0$ itself. This turns out to  be a not particularly trivial integral involving various special functions, but some techniques from random matrix theory and the theory of orthogonal polynomials allow us to express it exactly as a relatively simple finite sum in one variable, namely
\begin{equation}\label{eq:singsum}
  \mu_0=\frac{\eta+1}{(n+2)N}\sum_{j=0}^\eta\frac{M+2j+2}{(j+1)(M+j+1)}\left(N-\frac{(\eta-j)_{j+1}}{(N+1)_j}\right)^2
\end{equation}
where $M=n-2\eta,\ N=n-\eta+1$, and $(a)_{j+1}=a(a+1)\cdots(a+j)$ is the Pochhammer symbol. 
We can then bound Eq.~\eqref{eq:singsum} by elementary means to obtain $\mu_0\le (\eta+1)H_{\eta+1}$, establishing Theorem~\ref{thm:1}.
\end{proof}

We can show that the bound of Theorem~\ref{thm:1} is tight, in the following sense:
\begin{restatable}{thm}{thmtight}\label{thm:tight} 
Let $\rho=\Pi_\phi$ be an $\eta$-particle $n$-mode Slater determinant state, with $n=c\eta$ for some constant $c>2$. Then we have 
\begin{equation}
    \Var [\hat{\Pi}_{\phi}] \sim \left(\frac{c-1}{c}\right)(\eta+1)H_{\eta+1}\,,
\end{equation}
where ``$\sim$'' denotes  asymptotic equality as $n\to\infty$. 
\end{restatable}
That is, when taking $\rho$  to be a Slater determinant state of some fixed filling ratio $1/c$, the variance of the estimator corresponding to $\rho$ itself saturates the bound of Theorem~\ref{thm:1} up to a multiplicative constant which can be taken arbitrarily close to one by letting $c\to\infty$.  
Intuitively, the fact that the variance bound is saturated when estimating the overlap with a Slater determinant state equal to the unknown state itself   may be traced to the fact that the maximal eigenvalue of $\Gamma_R$ (recall Eq.~\eqref{eq:gammar}) is $\mu_0$, corresponding to the shell at zero distance from $\ketbra\phi$; concentrating $\rho$ on that shell (that is, setting it equal to $\ketbra \phi$) then produces the strongest test of the bound. Indeed, asymptotically analysing the expression Eq.~\eqref{eq:singsum} under the assumption $n=c\eta$  yields Theorem~\ref{thm:tight}. 

So far, we have focused on the sample-efficiency of our protocol. Also important, however, is the computational complexity of estimating the overlaps given a collection of shadows. We find that this too can be performed efficiently; indeed we have

\begin{restatable}{thm}{thmcomp}\label{thm:comp} 
Given a shadow $\hat\rho(u,Z)$ and an arbitrary Slater determinant state $\Pi_\phi$, the computational cost of calculating $\Tr[\hat\rho(u,Z)\Pi_\phi]$ scales as $\mco(n\eta^2)$. 
\end{restatable}

We emphasise that the (linear) $n$ dependence of Theorem~\ref{thm:comp} is unavoidable, as to even write down a general dense Slater determinant state takes time $\mco(n\eta)$. 

Next we turn to the task of estimating the expectation values of arbitrary particle-preserving quadratic fermionic observables. We begin with a bound on the sample complexity:

\begin{restatable}{thm}{thmham}\label{thm:ham} 
Let $\rho$ be an $\eta$-particle $n$-mode fermionic state, and let $h\in{\rm End}\,\mbc^n$ be a particle-preserving quadratic fermionic observable. Then we have 
\begin{equation}
    \Var [\hat{h}] \le (2\eta + 1)\|h_0\|_2^2\,,
\end{equation}
where $\|h_0\|_2^2$ is the (squared) 2-norm of the traceless component of $h$. 
\end{restatable}
Although still somewhat involved, the proof of Theorem~\ref{thm:ham} turns out to be considerably less demanding than that of the corresponding Theorem~\ref{thm:1}. At a high level, one can interpret this as a reflection of the differing ways that Slater determinants and particle-preserving quadratic operators interact with the representation theory of our setup. To see this, we begin by noting     that the measurement channel $\mcm:\eh\to\eh$ (recalling Eq.~\eqref{eq:chan}) is equivariant with respect to the natural action of $U(n)$ on $\eh$ (that is, $A\in\eh\mapsto U_\eta(u)AU_\eta(u)\ad$). Combined with the fact that (it turns out) $\eh$ decomposes multiplicity-freely with respect to that action of $U(n)$, it follows from Schur's lemma that $\mcm$ simply acts as a scalar on each $U(n)$-irrep in $\eh$~\cite{fulton1991representation}. Now, as we see in   Appendix~\ref{sec:ham}, the (traceless) quadratic particle-preserving   operators furnish exactly one of these irreps, so that their action under $\mcm$ is to simply be multiplied by a scalar (which we show in  Appendix~\ref{sec:ham} to be $1/(n+1)$). This is in stark contrast to the behaviour of the Slater determinants under the channel, which from Eq.~\eqref{eq:minvpr} we see to be significantly more complicated. In fact, having identified that they are eigenvectors of $\mcm^{-1}$ of eigenvalue $n+1$, an $\mco(n\|h_0\|_2^2)$ bound on the sample complexity of estimating our currently considered observables follows immediately from the techniques of Ref.~\cite{west2026classical}; this is also the scaling   attainable  from Ref.~\cite{wan2023matchgate}. The non-trivial content of Theorem~\ref{thm:ham} is then to reduce this from scaling  with $n$ to with $\eta$.

As in the case of Slater
determinants, we can by means of an explicit example show that our bound on the scaling is optimal:
\begin{restatable}{thm}{thmhamscaling}\label{thm:hamscaling} 
Let $n$ be even, and take $h_0=\id_{n/2}\oplus (-\id_{n/2})$ to be the diagonal observable consisting of $(n/2)$-many ones followed by $(n/2)$-many negative ones. Let $\ket\psi = (\ket{S_+}+\ket{S_-})/\sqrt 2$, where $S_+\subset[1,n/2]$ and $S_-\subset(n/2,n]$ are  $\eta$-particle subsets. Then we have 
\begin{equation}
    \Var [\hat{h}_0] = (n+1)\eta \in \Theta(\eta \|h_0\|_2^2)\,.
\end{equation}
\end{restatable}
Finally, we have that the estimation of these general particle-preserving quadratic observables also involves only efficient classical post processing:
\begin{restatable}{thm}{thmhamcomp}\label{thm:hamcomp} 
Given a shadow $\hat\rho(u,Z)$ and an arbitrary  particle-preserving quadratic observable $h$, the computational cost of calculating $\Tr[\hat\rho(u,Z)h]$ scales as $\mco(n^2\eta)$. 
\end{restatable}
As in the case of estimating  Slater determinants, we note that this (now increased to quadratic) $n$-dependence is optimal (assuming an  explicit  dense  input representation), on account of the fact that to even read  a generic such observable takes time $\mco(n^2)$. Happily, Theorem~\ref{thm:hamcomp} shows that the linear $\eta$-dependence is not increased from the Slater determinant case. 

\section{Discussion}
The primary technical contribution of this work is to derive, in the context of a certain classical shadow protocol, tight bounds on the scaling of the sample complexity for the estimation of several classes of observables. Indeed in classical shadows there is something of a history of it being the case  that it is relatively straightforward to derive non-trivial but quite loose bounds on the sample complexity~\cite{west2026classical}, with tight bounds requiring significantly more effort. For example, in the original proposal of Ref.~\cite{huang2020predicting}, it is noted that for local Clifford shadows one can quite easily obtain sample complexity bounds for a general $k$-local operator scaling as $9^k$; the improvement in that work down to $4^k$ requires significant technical cleverness. A similar situation occurs in the case of real local Cliffords~\cite{west2025real}. From a certain point of view, the heart of the difficulty is that the key contribution to the variance of the estimators is a \textit{third-order} quantity, 
\begin{equation}
\mbe[\hat{o}^2]=\sum_{w\in \mcw}\Tr[(\rho \ot \mcm^{-1}(O)\ts)\expect_{g\sim G}(g\Pi_w g\ad)\tt]\,,\label{eq:var}
\end{equation}
where we sample unitaries from some group $G$ (c.f. Eq.~\eqref{eq:gammar}).
In the context of the local Clifford example, the bound  of $9^k$ can be derived by making the simple estimate $\Tr [\mcm^{-1}(O) u\Pi_w u\ad]\le \|\mcm^{-1}(O)\|_\infty $ in Eq.~\eqref{eq:var}. Intuitively, one would expect this to be, on average over $u$, a rather weak assumption, so that it is not too surprising that the resulting bounds fail to be tight. Representation theoretically, the additional challenge of exactly evaluating  Eq.~\eqref{eq:var} compared to the measurement channel itself (Eq.~\eqref{eq:chan}) may be seen as a reflection of the increased   difficulty of working with  the trivial irreps of a tensor cube instead of a tensor square; the operator $\expect_{g\sim G}(g\Pi_w g\ad)\tk$, after all, projects onto the $k$-fold commutant of $G$~\cite{mele2024introduction}. Indeed, in the particle-preserving fermionic shadow case, it is precisely this difficulty that is cited in Ref.~\cite{low2022classical} as preventing the variances from being calculated exactly. 

In light of this, then, it is highly interesting to explore new techniques that can lead to improved control of the integrals that characterize the variance of classical shadow estimators. To that end, our employment of the theory of symmetric spaces joins a growing list of their technical applications to quantum computing and information theory. For example, and as mentioned during the above sketch, our proof of Theorem~\ref{thm:1} relies critically on invoking the theory of zonal spherical functions on the symmetric space $U(n)/(U(\eta)\times U(n-\eta))$. These functions recently played a role in the construction of exact unitary designs, by inductively constructing designs on $U(n)$ from designs on $U(\eta)\times U(n-\eta)$~\cite{bannai2022explicit}. Such recursive considerations have also lead to symmetric spaces playing an important role in compiling unitaries into simple components suitable for implementing on a physical quantum computer~\cite{wierichs2025recursive,dagli2008general}.  In an entirely different application, analysis of a \textit{finite} symmetric space associated to the symmetric group was used to prove a lower bound on the quantum query complexity of a certain quantum state generation problem~\cite{lindzey2019tight}. The random properties of unitaries drawn from symmetric spaces have also recently been the study of some explicit attention, with applications to information scrambling~\cite{hunter2018operator}, quantum learning theory~\cite{west2026average} and classical shadows~\cite{chang2026classical}. More generally, further exploring how the powerful techniques of analysis on symmetric spaces can be used in quantum computing is a natural line of inquiry.

Along with its sample complexity and  classical post-processing costs, there is in general a third critical ingredient to ensure  the efficiency of a classical shadow protocol: the depths of the circuits required to implement the random unitary sampling. While we have spoken of sampling uniformly from (a certain representation of) the unitary group $U(n)$, classical shadows in general probes only the third moment of the employed distribution~\cite{huang2020predicting}, so that a 3-design over the same suffices. In order to ascertain what that means exactly for us, we need for the first time to consider exactly how we are representing our states and unitaries at the circuit level. There are two natural choices: first and second quantization. 

In first quantization, we represent  an $\eta$-particle $n$-mode fermionic state $\ket\psi\in\Lm^\eta(\mbc^n)$ using $\eta\lceil\log_2 n\rceil$ qubits, with a valid state completely anti-symmetric under (for example) the exchange of the first $\log n$ qubits with the qubits at positions $j\log n$ through $(j+1)\log n$. Given a valid initial state, then, to apply a random $u\sim U(n)$ is in this representation  simply to apply $u^{\ot \eta}$, with an individual $u$ acting as the standard representation of the unitary group on its register of $\sim \log n$ qubits. To obtain a  3-design over this ensemble it is therefore clearly sufficient to sample $u$ from  a $3\eta$-design over $U(n)$. This can be done in depth $\widetilde{\mco}(\eta\log\log n)$, where the tilde hides polylog factors in $\eta$, and we work to a fixed approximation error~\cite{schuster2024random}. If one allows ancilla qubits, initial bits of randomness, and long range 2-qubit gates, this can be improved to $\mco(\log\eta\log\log(\eta\log n))$~\cite{cui2025unitary}, which is fairly good $n$-dependence.

Interestingly, the situation differs   considerably  in second quantization. Here  one works with $n$ qubits, with a state of definite particle number $\eta$ taking the form of a superposition over bitstrings of Hamming weight $\eta$, and random unitaries     given by (particle-preserving) matchgates. Now, the exact Haar measure on the group of (1D nearest-neighbour) particle-preserving matchgates can be sampled from in depth $\Theta(n)$~\cite{braccia2025optimal}, but, remarkably, designs in sublinear depth consisting of ensembles of (1D nearest-neighbour) particle-preserving matchgates constructed from local gates are not possible. This is due to a simple argument very similar to those recently made in the aid of ruling out  sublinear depth designs over various groups in Refs.~\cite{west2025no,grevink2025will,schuster2024random}. Indeed, consider applying such a particle-preserving matchgate, of depth $L$, to the state in which exactly the first $\eta$ modes are occupied. A simple lightcone argument reveals that any modes beyond position $\eta+L$ will remain unoccupied, behaviour that is readily distinguished from even a 1-design over the group. Having said that, it may yet be possible to construct sublinear depth designs by looking beyond ensembles of particle-preserving matchgates; the current situation, however, suggests that moving from first to second quantization trades improved qubit numbers for deeper circuits. It would of course be very interesting to resolve this conclusively.

\section{Acknowledgments}

MW acknowledges the support of the IBM Quantum Hub  at the University of Melbourne.
ML and MC acknowledge support by the Laboratory Directed Research and Development (LDRD) program of LANL under project number 20260043DR, and by the LANL's ASC Beyond Moore’s Law project. 
This work was also supported by the Quantum Science Center (QSC), a National Quantum Information Science Research Center of the U.S. Department of Energy (DOE). 

\bibliography{refs, quantum}

\newpage

\onecolumngrid
\appendix

\makeatletter
\renewcommand*{\theHequation}{\theHsection.\arabic{equation}}
\makeatother



\section{Estimating Slater determinant overlaps}\label{sec:sd}
In this Appendix we prove  Theorems~\ref{thm:1},~\ref{thm:tight} and~\ref{thm:comp}. 
We begin by recalling some notation. We let $\mcz_\eta=\{Z\subset[n]: |Z|=\eta\}$, with $ D=|\mcz_\eta|=\binom n\eta$. To simplify the notation, we will sometimes use $M=n-2\eta,\ N=n-\eta+1$. We will always take $\eta\le n/2$. 
For a given $Z=\{z_1<\cdots<z_\eta\}$, we define \textit{the occupation-basis vector} $|Z\rangle=e_{z_1}\wedge\cdots\wedge e_{z_\eta}
  \in\mch_\eta:=\Lambda^\eta(\C^n)$, and call $R=\{1,2,\ldots,\eta\}$ the \textit{reference occupation}. 
For any $Z\in\mcz_\eta$, we write $\Pi_Z =\ketbra{Z}$, and define the \textit{Johnson distance} between two $\eta$-subsets to be $d(S,Z)=\eta-|S\cap Z|$. Following the notation of Ref.~\cite{low2022classical}, we use $U_\eta:  U(\C^n)\to  U(\mch_\eta)$ to denote the exterior power representation $U_\eta(u)=\Lambda^\eta u$, and $dU_\eta: \mfu(\C^n)\to \mfu(\mch_\eta)$ for its derivative. In the occupation basis one then has the standard formulae
\begin{equation}
  U_\eta(u)|Z\rangle=\sum_{T\in\mcz_\eta}\det(u_{T,Z})|T\rangle \,,
\label{eq:minors}
\end{equation}
where $u_{T,Z}$ is the $\eta\times\eta$ submatrix of $u$ with rows indexed by $T$ and columns indexed by $Z$, and 
\begin{equation}
    dU_\eta(h) \ket Z = \sum_{r=1}^\eta e_{z_1}\wedge\cdots\wedge(he_{z_r})\wedge\cdots\wedge e_{z_\eta}\,;
\end{equation}
equivalently, $dU_\eta(h) = \sum_{i,j}h_{i,j}a_i\ad a_j$. We will abuse notation slightly to apply $dU_\eta$ to hamiltonians $h\in {\rm End}\,\mbc^n$; as $\mfu(\C^n)$ consists of skew-hermitian matrices one would more properly write $dU_\eta(ih)=i\sum_{i,j}h_{i,j}a_i\ad a_j$.
We also  define, for $0\le t\le\eta$, the shell projector
\begin{equation}
  P_t=\sum_{\substack{T\in\mcz_\eta\\ d(T,R)=t}}\Pi_T.
\label{eq:pt}
\end{equation}
The shell size is given by $N_t=\Tr(P_t)=\binom{\eta}{t}\binom{n-\eta}{t}$; indeed, an element at a Johnson distance of $t$ from $R$ is specified exactly by choosing which $t$ elements of $R$ to remove, and with which $t$ elements  which to replace them.  Now, the classical shadow channel corresponding to our protocol is
\begin{equation}\label{eq:chanapp}
    \mcm (\rho) = \int_{u\sim  U(n)} \sum_{Z\in\mcz_\eta} \Tr[\rho U_\eta(u)\ad \Pi_ZU_\eta(u)]U_\eta(u)\ad \Pi_Z U_\eta(u)\,; 
\end{equation}
as discussed in the main text, we need to know its (inverse) action on an arbitrary Slater determinant state $\Pi_Z$. Indeed, our estimators are given by $\widehat\Pi_\phi (u,Z)= \Tr[\Pi_\phi \mcm^{-1}(U_\eta(u)\ad\Pi_ZU_\eta(u))] $. To that end, we have:
\begin{restatable}{lem}{leminv}\label{lem:inv} 
For every $Z\in\mcz_\eta$,
\begin{equation}
  \mcm^{-1}(\Pi_Z)=
  \sum_{t=0}^\eta \beta_t
\sum_{\substack{T\in\mcz_\eta\\ d(T,Z)=t}}\Pi_T
   \,,
\label{eq:inv}
\end{equation}
where $\beta_t=(-1)^t\frac{\binom{n-t}{\eta-t}}{\binom{\eta}{t}}$.
\end{restatable}
\begin{proof}
First,  by the manifest equivariance of the channel Eq.~\eqref{eq:chanapp}, $\mcm^{-1}(U_\eta(u)^\dagger\Pi_ZU_\eta(u))=U_\eta(u)^\dagger\mcm^{-1}(\Pi_Z)U_\eta(u)$, it clearly suffices to establish the result for the reference state $R=[\eta]$. Indeed, imagine that Eq.~\eqref{eq:inv} holds for the reference state. Then, for any  $Z=v_ZR$ with $v_Z$ a permutation unitary, we would have
\begin{align}
     \mcm^{-1}(\Pi_Z)&=\mcm^{-1}(U_\eta(v_Z)\Pi_RU_\eta(v_Z)\ad)\\
     &=U_\eta(v_Z)\mcm^{-1}(\Pi_R)U_\eta(v_Z)\ad\\
     &=\sum_{t=0}^\eta \beta_t\sum_{\substack{T\in\mcz_\eta\\ d(T,R)=t}}U_\eta(v_Z)\Pi_TU_\eta(v_Z)\ad\\
     &=\sum_{t=0}^\eta \beta_t\sum_{\substack{T\in\mcz_\eta\\ d(v_Z^{-1}T,R)=t}}\Pi_T=\sum_{t=0}^\eta \beta_t\sum_{\substack{T\in\mcz_\eta\\ d(T,v_ZR)=t}}\Pi_T=\sum_{t=0}^\eta \beta_t\sum_{\substack{T\in\mcz_\eta\\ d(T,Z)=t}}\Pi_T\,.
\end{align}
So, let $Z=R$. For $\eta$-particle occupation-number states $P,Q$, the single-shot estimator of an operator $O=\sum_{P,Q}o_{P,Q}\ketbra{P}{Q}$ corresponding to a sample $(u,Z)$ is  given in Ref.~\cite{low2022classical} as
\begin{equation}
    \langle \hat O \rangle = \Tr[O U_\eta(u)\ad \mcm^{-1}(\Pi_Z)U_\eta(u)] =  \Tr[O U_\eta(v_Z\ad u)\ad E_{\eta,\eta} U_\eta(v_Z\ad  u)]\,,
\end{equation}
where 
\begin{equation}
    E_{\eta,\eta} = \sum_{S\in\mcz_\eta} \ketbra{S}{S}(-1)^{\eta+s'} \frac{\binom{n-\eta+s'}{s'}}{\binom{\eta}{s'}}\,,
\end{equation}
with $s'=|S\cap [\eta]|$. Noticing $s'=\eta-d(S,[\eta])$ concludes the proof. 
\end{proof}
Now, by the manifest equivariance of the channel Eq.~\eqref{eq:chanapp}, $\mcm^{-1}(U_\eta(u)^\dagger\Pi_ZU_\eta(u))=U_\eta(u)^\dagger\mcm^{-1}(\Pi_Z)U_\eta(u)$, we have, for any outcome pair $(u,Z)$ and any permutation unitary $v_Z$ with $v_ZR=Z$, 
\begin{align}
\widehat\Pi_R(u,Z)&=\Tr\left[\Pi_RU_\eta(u)^\dagger \mcm^{-1}(\Pi_Z)U_\eta(u)\right]\\
  &=\Tr\left[\Pi_RU_\eta(u)^\dagger \mcm^{-1}(U_\eta(v_Z)\Pi_RU_\eta(v_Z\ad))U_\eta(u)\right]\\
  &=\Tr\left[\Pi_RU_\eta(u)^\dagger U_\eta(v_Z)\mcm^{-1}(\Pi_R)U_\eta(v_Z\ad)U_\eta(u)\right]\\
  &=\Tr\left[\Pi_RU_\eta(v_Z\ad u)^\dagger \mcm^{-1}(\Pi_R)U_\eta(v_Z\ad u)\right]\,. 
\end{align}
Defining
\begin{equation}\label{eq:h}
    F(w)=\Tr\left[\Pi_RU_\eta(w)^\dagger \mcm^{-1}(\Pi_R)U_\eta(w)\right],
\end{equation}
we therefore have $\widehat\Pi_R(u,Z)=F(v_Z\ad u)$. The function $F$ has a reasonably explicit form; indeed we find
\begin{restatable}{lem}{lemh}\label{lem:hlem} 
 $F(w)=g_\eta(w_{R,R}^\dagger w_{R,R})$, where
\begin{equation}
  g_\eta(X)=\sum_{r=0}^\eta(-1)^{\eta-r}\frac{\binom{n+1}{r}}{\binom{\eta}{r}}e_r(X)\,,
\label{eq:g}
\end{equation}
and $e_r(X)$ denotes the $r$\textsuperscript{th} elementary symmetric polynomial in the eigenvalues of the matrix $X$, with $e_0(X)=1$.
\end{restatable}
\begin{proof}
By Lemma~\ref{lem:inv} we  have $ F(w)=\sum_{t=0}^\eta\beta_tq_t(w)$, where 
\begin{equation}\label{eq:qdef}
q_t(w)=\Tr\left[P_tU_\eta(w)\Pi_RU_\eta(w)^\dagger\right]\,,    
\end{equation}
and $P_t=\sum_{T\in\mcz_\eta\,|\,d(T,R)=t}\Pi_T$ is the projector onto the shell at distance $t$ from $R$. Our first observation is that, with $X=w_{R,R}^\dagger w_{R,R}$, we have the generating function
\begin{equation}
  \sum_{t=0}^\eta q_t(w)y^{\eta-t}=\det(\id_\eta-X+yX)\,.\label{eq:gen}
\end{equation}
Let us prove Eq.~\eqref{eq:gen}. To do so, we make our first (but not last) explicit appeal to a property of the AIII symmetric space $G/K= U(n)/( U(\eta)\times U(n-\eta))$. Allowing $K$ to act in the natural way (that is, so as to preserve the splitting $ \C^n={\rm span}\{e_1,\ldots,e_\eta\}\oplus{\rm span}\{e_{\eta+1},\ldots,e_n\}$) we see that for $k\in K$ we have both that  $U_\eta(k)P_tU_\eta(k)^\dagger=P_t$ and that  $U_\eta(k)|R\rangle$ is simply a phase times $|R\rangle$, because the first block of $k$ acts on $e_1\wedge\cdots\wedge e_\eta$ by its determinant. It follows that $q_t(k_1wk_2)=q_t(w)$  
for $k_1,k_2\in K$, and so $q_t(w)$ depends only on the double coset $KwK$. We can therefore use the $KaK$ theorem on AIII to without loss of generality assume that $w$ is in the ``cosine-sine form'', so that there exist   orthonormal vectors
$e_1',\ldots,e_\eta'\in R$, $f_1',\ldots,f_\eta'\in R^\perp$, such that  $w e_i'=\sqrt{\lambda_i}\,e_i'+\sqrt{1-\lambda_i}\,f_i'$ for $i=1,\ldots,\eta$, and where   the numbers $\lambda_i\in[0,1]$ are the eigenvalues of $X$. We then have
\begin{equation}\label{eq:wwedge}
    U_\eta(w)|R\rangle=\bigwedge_{i=1}^\eta\left(\sqrt{\lambda_i}e_i'+\sqrt{1-\lambda_i}f_i'\right)\,.
\end{equation}
Now, expanding out the brackets in Eq.~\eqref{eq:wwedge} and inserting the result into
Eq.~\eqref{eq:qdef},  we see that $P_t$ will kill any term in the expansion that did not involve ``choosing the $f$ term'' exactly $t$ times, and act as the identity on those terms which do, whence we obtain
\begin{equation}
  q_t(w)=\sum_{\substack{I\subset\{1,\ldots,\eta\}\\ |I|=t}}\prod_{i\in I}(1-\lambda_i)\prod_{i\notin I}\lambda_i.
\end{equation}
Multiplying by $y^{\eta-t}$ and summing over $t$ then yields 
\begin{align}
\sum_{t=0}^\eta q_t(w)y^{\eta-t}&=\sum_{t=0}^\eta y^{\eta-t} \sum_{\substack{I\subset\{1,\ldots,\eta\}\\ |I|=t}}\prod_{i\in I}(1-\lambda_i)\prod_{i\notin I}\lambda_i\\
&=\sum_{t=0}^\eta \sum_{\substack{I\subset\{1,\ldots,\eta\}\\ |I|=t}}\prod_{i\in I}(1-\lambda_i)\prod_{i\notin I}y\lambda_i\\
&=  \sum_{ I\subset\{1,\ldots,\eta\} }\prod_{i\in I}(1-\lambda_i)\prod_{i\notin I}y\lambda_i\\
  &=\prod_{i=1}^\eta\bigl((1-\lambda_i)+y\lambda_i\bigr)\\
  &=  \det(\id_\eta-X+yX)\,,
\end{align}
which is Eq.~\eqref{eq:gen}. Now, from the readily verified identity $  \det(\id+zX)=\sum_{r=0}^\eta e_r(X)z^r$, where $e_r(X)$ denotes the $r$\textsuperscript{th} elementary symmetric polynomial in the eigenvalues of the matrix $X$, we thus have
\begin{align}
    \sum_{t=0}^\eta q_t(w)y^{\eta-t}&=\sum_{r=0}^\eta e_r(X)(y-1)^r\\
    &=\sum_{r=0}^\eta e_r(X)\sum_{a=0}^r\binom r a y^a(-1)^{r-a}\\
    &=\sum_{a=0}^\eta\sum_{r=a}^\eta e_r(X)\binom r a y^a(-1)^{r-a}\\
    &=\sum_{t=0}^\eta\left[\sum_{r=\eta-t}^\eta e_r(X)\binom {r} {\eta-t} (-1)^{r-\eta-t}\right]y^{\eta-t}
\end{align}
so that $q_t(w)=\sum_{r=\eta-t}^\eta e_r(X)\binom {r} {\eta-t} (-1)^{r-\eta-t}$. Continuing on, we have
\begin{align}
    F(w)&=\sum_{t=0}^\eta\beta_tq_t(w)\\
    &=\sum_{t=0}^\eta \left((-1)^t\frac{\binom{n-t}{\eta-t}}{\binom{\eta}{t}}\right)\left[\sum_{r=\eta-t}^\eta e_r(X)\binom {r} {\eta-t} (-1)^{r-\eta-t}\right] \\
    &=\sum_{s=0}^\eta\sum_{r=s}^\eta (-1)^{r-\eta} \frac{\binom{n-\eta+s}{s}}{\binom{\eta}{s}}  \binom {r} {s}  e_r(X)\\
    &=\sum_{r=0}^\eta\sum_{s=0}^r (-1)^{r-\eta} \frac{\binom{n-\eta+s}{s}}{\binom{\eta}{s}}  \binom {r} {s}  e_r(X)\\
    &=\sum_{r=0}^\eta\sum_{s=0}^r (-1)^{r-\eta} \frac{\binom{n-\eta+s}{s}}{\binom{\eta}{r}}  \binom {\eta-s} {\eta-r}  e_r(X)
\end{align}
where we made the change of variables $s=\eta-t$. So we will be done if we can show
\begin{equation}\label{eq:pregen}
    \sum_{s=0}^r \binom{n-\eta+s}{s}  \binom {\eta-s} {\eta-r} = \binom{n+1}{r}\,,
\end{equation}
which one can do via a relatively painless generating function argument. Indeed, the left hand side of Eq.~\eqref{eq:pregen} is exactly the coefficient of $z^r$ in the series
\begin{align}
  \sum_{s\ge0}\binom{n-\eta+s}{s}z^s(1+z)^{\eta-s}
  &=(1+z)^\eta\sum_{s\ge0}\binom{n-\eta+s}{s}
  \left(\frac{z}{1+z}\right)^s  \\
  &=(1+z)^\eta\left(1-\frac{z}{1+z}\right)^{-(n-\eta+1)}\\
  &=(1+z)^{n+1}\\
  &=\sum_{r=0}^{n+1}\binom{n+1}{r}z^r\,,
\end{align}
and (having used the identity $  \sum_{s\ge0}\binom{a+s}{s}x^s=(1-x)^{-a-1}$~\cite{wilf2005generatingfunctionology}) we are done. 
\end{proof}
We now turn directly to the consideration of the variance of our estimators. We define the second moment operator
\begin{equation}\label{eq:gamma}
    \Gamma_R=\int_{u\sim  U(n)}\sum_{Z\in\mcz_\eta}
\widehat\Pi_R(u,Z)^2U_\eta(u)^\dagger\Pi_ZU_\eta(u)\, ;
\end{equation}
as explained in the main text, we have $\Var [\hat{\Pi}_{Z}] \le \|\Gamma_R\|_\infty$, and our aim is then to characterise the largest eigenvalue of $\Gamma_R$. To that end, we have the following lemma: 
\begin{restatable}{lem}{lemgamma}\label{lem:gamma} 
 The spectral decomposition of $\Gamma_R$ reads $\Gamma_R\cong \bigoplus_{s=0}^\eta \mu_s\id_{\mch_s}$, where $\mch_s={\rm span}\,\{|S\rangle:d(S,R)=s\}$ is the \textit{Johnson shell} at distance $s$ from $R$.
\end{restatable}
\begin{proof}
Consider the natural exterior action of the group $K= U(\eta)\times U(n-\eta)$ on $\mch_\eta$. This action commutes with $\Gamma_R$; indeed, for $k\in K$ we may calculate:
\begin{align}
    U_\eta(k)\Gamma_RU_\eta(k)^\dagger&=U_\eta(k) \left[\int_{u\sim  U(n)}\sum_{Z\in\mcz_\eta}\widehat\Pi_R(u,Z)^2U_\eta(u)^\dagger\Pi_ZU_\eta(u)\right]U_\eta(k)^\dagger\\
    &=\int_{u\sim  U(n)}\sum_{Z\in\mcz_\eta}\widehat\Pi_R(u,Z)^2U_\eta(uk\ad)^\dagger\Pi_ZU_\eta(uk\ad)\\
    &=\int_{u\sim  U(n)}\sum_{Z\in\mcz_\eta}\widehat\Pi_R(uk,Z)^2U_\eta(u)^\dagger\Pi_ZU_\eta(u)\\
    &=\int_{u\sim  U(n)}\sum_{Z\in\mcz_\eta}\widehat\Pi_R(u,Z)^2U_\eta(u)^\dagger\Pi_ZU_\eta(u)\\
    &=\Gamma_R
\end{align}
where we have used the invariance  of the Haar measure on $ U(n)$ and that 
\begin{align}
    \widehat\Pi_R(uk,Z)  &= \Tr[\Pi_R \mcm^{-1}(U_\eta(uk)\ad\Pi_ZU_\eta(uk))]\\
     &= \Tr[\Pi_R \mcm^{-1}(U_\eta(k)\ad U_\eta(u)\ad\Pi_ZU_\eta(u)U_\eta(k))]\\
    &= \Tr[U_\eta(k)\Pi_RU_\eta(k\ad) \mcm^{-1}(U_\eta(u)\ad\Pi_ZU_\eta(u))]\\
    &=\Tr[\Pi_R \mcm^{-1}(U_\eta(u)\ad\Pi_ZU_\eta(u))]\\
    &= \widehat\Pi_R(u,Z)\, ,
\end{align}
which follows from the fact that for $k\in K$ we have $U_\eta(k)\ket R = \det (k_{R,R}) \ket R$ (which follows in turn from Eq.~\eqref{eq:minors}). So, by Schur's lemma, $\Gamma_R$ will act non-trivially only on the multiplicity spaces of the decomposition of $\mch_\eta$ into $K$-irreps. As the action of $K$ cannot change the Johnson distance of an occupation basis vector from the reference state (that is, if  $Z$ satisfies $d(Z,R)=s$, then $K\ket Z$ is a linear combination of basis vectors which are also at distance $s$), we at least have the decomposition $\mch_\eta \cong_K \bigoplus_{s=0}^\eta \mch_s$, where $\mch_s={\rm span}\,\{|S\rangle:d(S,R)=s\}$. These subspaces turn out to not decompose further, but are rather (pairwise inequivalent) irreps of $K$. Indeed, $K$ acts on  $\mch_s $ as $\Lambda^{\eta-s}( U(\mbc^\eta))\otimes \Lambda^{s}( U(\mbc^{n-\eta}))$. As each exterior power $\Lambda^r(\mbc^m)$ is an irrep of $ U(m)$~\cite{fulton1991representation}, and as the tensor product of irreps is itself an irrep with respect to the action of the corresponding product group, we see that each $\mch_s$ is irreducible as a $K$-representation. This establishes the claim of the lemma. 
\end{proof}

We find it useful to introduce (writing $G= U(n)$) the functions $p_s: G/K\to \mbr$ given by
\begin{equation}\label{eq:ps}
    p_s(gK)=\frac{1}{N_s}
  \sum_{\substack{S\in\mcz_\eta\\d(S,R)=s}}|\langle S|U_\eta(g)|R\rangle|^2\,,
\end{equation}
i.e. the average  transition probability from $R$ to $S$ under $g$. For example, we have $p_0(gK)=\lvert\langle R|U_\eta(g)|R\rangle\rvert^2=\lvert\det(g_{R,R})\rvert^2$. The fact that $p_s$ is well-defined as a function on $G/K$ follows from the previously seen $U_\eta(k)\ket R = \det(k_{R,R})\ket R$.  We will need:
\begin{restatable}{lem}{lempsinv}\label{lem:psinv} 
 The functions $p_s$ satisfy $  p_s(gK)=p_s(g^{-1}K)$ for all $g\in G$. 
\end{restatable}
\begin{proof}
Let us write $g$ in block form relative to $\C^n=R\oplus R^\perp$ as $g=\begin{pmatrix}A&B\\ C&D\end{pmatrix}$, 
where (e.g.) $A$ is the $\eta\times\eta$ block with rows and columns in $R$, and collect the (unnormalised) shell sums into the generating polynomial
\begin{equation}
  Q_g(z)=\sum_{s=0}^\eta N_sp_s(g)z^s =\sum_{|S|=\eta}z^{|S\cap R^\perp|}\lvert\det(g_{S,R})\rvert^2\,.
\end{equation}
We also introduce the diagonal matrix $W_z=\begin{pmatrix}\id_R&0\\0&z\id_{R^\perp}\end{pmatrix}$. Now, recall that for any $\eta \times n$ matrix $X$ and $n \times \eta$ matrix $Y$, the Cauchy--Binet formula states 
\begin{equation}\label{eq:cb}
    \det(XY) = \sum_{S\in\binom{[n]}{\eta}} \det(X_{[\eta],S})\det(Y_{S,[\eta]})\,.
\end{equation}
Defining the $n\times \eta$ matrix $g_{*,R}$ containing the first $\eta$ columns of $g$ and applying Eq.~\eqref{eq:cb} to the product $g_{*,R}^\dagger (W_zg_{*,R})$ then gives
\begin{align}
  Q_g(z)&=\det(g_{*,R}^\dagger W_zg_{*,R})\\
  &=\det(A^\dagger A+zC^\dagger C  )\\
  &=\det(A^\dagger A+z (\id_\eta-A^\dagger A)  )
\end{align}
where we have used the unitarity of $g$.
Now, in the case of $g^{-1}$, the same formula gives
\begin{equation}
  Q_{g^{-1}}(z)=\det(AA^\dagger+z(\id_\eta-AA^\dagger))\,.
\end{equation}
But the matrices $A^\dagger A$ and $AA^\dagger$ have the same eigenvalues, so that $Q_{g^{-1}}(z)=Q_g(z)$, and therefore  $p_s(g^{-1})=p_s(g)$, which is what we wanted to prove.
\end{proof}
Our next goal is to express $\mu_s$ as an integral over $G/K$. With $F$ as in Eq.~\eqref{eq:h}, we find:
\begin{restatable}{lem}{lemmuint}\label{lem:muint} 
 For every $s$,
\begin{equation}
  \mu_s=\binom n\eta\int_{x\sim G/K}F(x)^2p_s(x)\,.
\label{eq:mus}
\end{equation}
\end{restatable}
\begin{proof}
To begin, we know from Lemma~\ref{lem:gamma} that the eigenvalue $\mu_s$ may be extracted by contracting $\Gamma_R$ with the projector $P_s = \sum_{S\in\mcz_\eta\,:\,d(S,R)=s}\ketbra{S}$ onto the corresponding eigenspace, namely via $\mu_s=N_s^{-1}\Tr[P_s\Gamma_R]$. Substituting in our expression Eq.~\eqref{eq:gamma} for $\Gamma_R$ then gives
\begin{align}
 \mu_s&=\frac1{N_s}   \int_{u\sim  U(n)}\sum_{Z\in\mcz_\eta}\widehat\Pi_R(u,Z)^2\Tr[P_s U_\eta(u)^\dagger\Pi_ZU_\eta(u)]\\
 &=\frac1{N_s}   \int_{u\sim  U(n)}\sum_{Z\in\mcz_\eta}\widehat\Pi_R(u,Z)^2\sum_{\substack{S\in\mcz_\eta\\d(S,R)=s}}\lvert\braket{Z|  U_\eta(u)}{S}\rvert^2
\end{align}
Now, for every fixed $Z$, choose a permutation unitary $v_Z$ with $v_ZR=Z$, and let $v=v\ad_Z u$. By Eq.~\eqref{eq:h}, $\widehat\Pi_R(u,Z)=F(v)$, and by the invariance of the Haar measure on $ U(n)$ under $u\mapsto v\ad_Z u=v$ we have
\begin{align}
    \mu_s &= \frac1{N_s} \sum_{Z\in\mcz_\eta}  \int_{v\sim  U(n)}F(v)^2\sum_{\substack{S\in\mcz_\eta\\d(S,R)=s}}\lvert\braket{R|  U_\eta(v)}{S}\rvert^2\\
    &=  \sum_{Z\in\mcz_\eta}  \int_{v\sim  U(n)}F(v)^2p_s(v^{-1}K) \\
    &=  \sum_{Z\in\mcz_\eta}  \int_{v\sim  U(n)}F(v)^2p_s(vK) \\
    &=  \binom{n}{\eta} \int_{v\sim  U(n)}F(v)^2p_s(vK) 
\end{align}
where we have used Lemma~\ref{lem:psinv}. Finally,  both $F(v)$ and $p_s(vK)$ are right $K$-invariant, so that the integral over $U(n)$ descends to the integral over $G/K$.  
\end{proof}
So, our task is to bound $\max_{0\le s \le \eta} \mu_s$, with $\mu_s$ given by Eq.~\eqref{eq:mus}. This is  greatly facilitated by the following lemma:
\begin{restatable}{lem}{lemmuzero}\label{lem:muzero} 
 The scalars $\mu_s$ satisfy $ \mu_0 =\max_{0\le s \le \eta} \mu_s$. 
\end{restatable}
The proof of Lemma~\ref{lem:muzero} requires something of a diversion through the theory of \textit{zonal spherical functions}, and is relegated to Appendix~\ref{sec:zsf}. The remaining task, then, is to evaluate $\mu_0=\binom{n}{\eta}\int_{x\sim G/K}F(x)^2p_0(x)$. 

Let us introduce, for $u\in U(n)$  Haar random, the random matrix $X=u_{R,R}\ad u_{R,R}$. The matrix $X$  is then distributed according to the complex matrix beta type I distribution, $X\sim {\rm Beta}_\eta^{\rm I}(\eta,n-\eta)$~\cite{kieburg2016singular,mathai2022matrix}; multiplying the corresponding density by $\binom{n}{\eta} p_0(uK)=\binom{n}{\eta} \det X$ produces the distribution $X\sim {\rm Beta}_\eta^{\rm I}(\eta+1,n-\eta)$~\cite{mathai2022matrix}. 
By Lemma~\ref{lem:hlem} and Eq.~\eqref{eq:ps} we then have
\begin{align}
    \mu_0 &= \binom{n}{\eta}\int_{u\sim  U(n)}g_\eta(X)^2\det X\\
    &= \expect_{X\sim {\rm Beta}_\eta^{\rm I}(\eta+1,n-\eta)}[g_\eta(X)^2]\label{eq:muexp}
\end{align}
Now, the eigenvalues of $X$ under the det-biased measure are distributed as~\cite{kieburg2016singular}
\begin{equation}\label{eq:mario}
    p(x_1,\ldots x_\eta) \propto \prod_{i<j}(x_i-x_j)^2\prod_j x_j(1-x_j)^{n-2\eta} {\rm d} x_j = \Delta(x)^2\prod_j x_j(1-x_j)^{n-2\eta} {\rm d} x_j\,,
\end{equation}
where we have introduced the Vandermonde determinant $\Delta$. This is the Jacobi unitary ensemble with weight $w(x)=x(1-x)^{M}$, where we introduce $M=n-2\eta$. To make further progress, we  turn to the theory of orthogonal polynomials. 

We can define an inner product on the functions on $[0,1]$ with respect to $w$, namely $\braket{f}{g}_w=\int_{0}^{1}f(x)g(x)w(x){\rm d}x$. Somewhat remarkably, the \textit{shifted Jacobi polynomials}
\begin{equation}
    J_j(x) = (j+1){}_2F_1(-j, j+M+2;2;x)
\end{equation}
are orthogonal with respect to this inner product; specifically we have
\begin{equation}\label{eq:jnorm}
    \int_{0}^{1}J_j(x)J_k(x)w(x){\rm d}x = \d_{jk} \frac{j+1}{(M+2j+2)(M+j+1)}\,.
\end{equation}
Let us define also, for $d\ge 0$, the $(d+1)\times(d+1)$ moment matrix for the weight $w(x)=x(1-x)^M$, with matrix elements
\begin{equation}
  H^{(d)}_{ab}=\int_0^1x^{a+b}w(x)\,{\rm d} x
  =B(a+b+2,M+1)\,,
\label{eq:mommatrix}
\end{equation}
where $B(\cdot,\cdot)$ is the beta function.  We have:
\begin{restatable}{lem}{lemmulin}\label{lem:mulin} 
Let $\mcl$ be the linear functional   on polynomials of degree at most $\eta$ that acts as $\calL(x^m)=
  \binom{n+1}{\eta-m}/\binom{\eta}{m}
  =:\ell_m$. Then
\begin{equation}\label{eq:mulem}
    \mu_0 = \frac{\det H^{(\eta)}}{\det H^{(\eta-1)}} \sum_{j=0}^\eta\frac{\calL(J_j)^2}{\|J_j\|^2}\,.
\end{equation} 
\end{restatable}
\begin{proof}
    Our  starting point is that from Eqs.~\eqref{eq:muexp} and~\eqref{eq:mario} we have
\begin{equation}\label{eq:muratio}
      \mu_0=\frac{\int_{[0,1]^\eta}g_\eta(x)^2\Delta(x)^2\prod_i x_i(1-x_i)^{n-2\eta} {\rm d} x_i}{\int_{[0,1]^\eta}\Delta(x)^2\prod_i x_i(1-x_i)^{n-2\eta}{\rm d} x_i}\,.
\end{equation}
We will begin by analysing the product $g_\eta(x)\Delta(x)$. To that end, we define the \textit{alternant} $A_m(x)=\det[x_i^a]_{1\le i\le \eta;\,0\le a\neq m \le \eta}$. For example, the Vandermonde determinant is equivalently written $\Delta=A_\eta$, and more generally we have  $  e_{\eta-m}(x)\Delta(x)=A_m(x)$. Indeed, introducing some formal variable $z$ one sees:
\begin{equation}
    \sum_{m=0}^\eta z^me_m(x) \Delta(x) = \left[\prod_{m=1}^\eta (1+z x_m)\right] \Delta(x) = \det[x_i^a+zx_i^{a+1}]_{1\le i\le \eta;\,0\le a < \eta}
\end{equation}
Now, consider expanding out the determinant using  column multilinearity. For a particular term in the expansion, one for each column ``chooses'' either the original column or the shifted column. Having chosen, say, the $j$\textsuperscript{th} column to be shifted, however, to obtain a non-zero contribution to the determinant all subsequent columns must also be chosen to be shifted, lest two columns be proportional. We obtain therefore $\eta+1$ terms, corresponding exactly to the $\eta+1$ choices of ``first shifted column''. This immediately yields $\det[x_i^a+zx_i^{a+1}]_{1\le i\le \eta;\,0\le a < \eta}=\sum_{m=0}^\eta z^mA_{\eta-m}(x)$, so that indeed $  e_{\eta-m}(x)\Delta(x)=A_m(x)$. From Eq.~\eqref{eq:g} we therefore have
\begin{equation}\label{eq:ga}
  g_\eta(x)\Delta(x)=\sum_{r=0}^\eta(-1)^{\eta-r}\frac{\binom{n+1}{r}}{\binom{\eta}{r}}e_r(x)\Delta(x)=\sum_{m=0}^\eta(-1)^{m}\frac{\binom{n+1}{\eta-m}}{\binom{\eta}{m}}A_{m}(x)=\sum_{m=0}^\eta(-1)^{m}\ell_mA_{m}(x)\,,
\end{equation}
where we have defined $\ell_m=\binom{n+1}{\eta-m}/\binom{\eta}{m}$.

Now, we can relate the denominator of Eq.~\eqref{eq:muratio} to the determinant of the moment matrix Eq.~\eqref{eq:mommatrix} by recalling \textit{Andreief’s identity}, which states for sequences of integrable functions $f_j,\,g_j$ that
\begin{equation}
    \int \det[f_j(x_k)]\det[g_j(x_k)]\prod_{j=1}^\eta {\rm d}\mu(x_j) = \eta!\det\int f_j(x)g_k(x){\rm d}\mu(x)
\end{equation}
(which one can choose to interpret as a sort of continuous version of the Cauchy–Binet formula Eq.~\eqref{eq:cb}). Indeed, taking $f_j(x)=g_j(x)=x^j$ and ${\rm d}\mu(x)=x(1-x)^{n-2\eta}{\rm d}x$ we see that the denominator of Eq.~\eqref{eq:muratio} becomes 
\begin{equation}
    \int_{[0,1]^\eta}\Delta(x)^2\prod_i x_i(1-x_i)^{n-2\eta}{\rm d} x_i = \eta! \det \int_{[0,1]} x^{j+k}  x(1-x)^{n-2\eta}{\rm d} x=\eta!\det H^{(\eta-1)}\,,
\end{equation}
where we have recalled Eq.~\eqref{eq:mommatrix}. For the numerator, Eq.~\eqref{eq:ga} similarly yields
\begin{align}
    \int_{[0,1]^\eta}g_\eta(x)^2\Delta(x)^2\prod_i x_i(1-x_i)^{n-2\eta} {\rm d} x_i &=\int_{[0,1]^\eta}\left(\sum_{m=0}^\eta(-1)^{m}\ell_mA_{m}(x)\right)^2\prod_i x_i(1-x_i)^{n-2\eta} {\rm d} x_i\\
&=\sum_{m=0}^\eta\sum_{l=0}^\eta(-1)^{m+l}\ell_m\ell_l\int_{[0,1]^\eta} A_{m}(x)A_{l}(x)\prod_i x_i(1-x_i)^{n-2\eta} {\rm d} x_i\\
&=\eta!\sum_{m=0}^\eta\sum_{l=0}^\eta(-1)^{m+l}\ell_m\ell_l\det H^{(\eta)}_{\widehat{m},\widehat{l}} \\
&=\eta!\sum_{m=0}^\eta\sum_{l=0}^\eta \det\big(H^{(\eta)}\big)\ell_m\ell_l  (H^{(\eta)})^{-1}_{ l,m} \\
&=\eta! \det\big(H^{(\eta)}\big)\ell^\mst  (H^{(\eta)})^{-1} \ell
\end{align}
where we have used Jacobi's cofactor formula for the matrix elements of an inverse matrix, and that the alternants may equivalently be defined as $A_m(x)=\det [\tilde{f}_j(x_k)]$, where 
\begin{equation}
    \tilde{f}_j(x) = \begin{cases}
        x^j&j<m\\
        x^{j+1}&j\ge m\\
    \end{cases}\,.
\end{equation}
To summarise our progress to this point, then, we have in the pursuit of Eq.~\eqref{eq:mulem} established
\begin{equation}
    \mu_0 = \frac{\det H^{(\eta)} }{\det H^{(\eta-1)} } \ell^\mst  (H^{(\eta)})^{-1} \ell\,,
\end{equation}
and would therefore like to further show $\ell^\mst  (H^{(\eta)})^{-1} \ell=\sum_{j=0}^\eta {\calL(J_j)^2}/{\|J_j\|^2}$; let us now see it. Indeed, recall that $H^{(\eta)}$ is the Gram matrix with respect to  the monomial basis $1,x,\ldots,x^\eta$ for the inner product $  \langle f,g\rangle=\int_{x\sim [0,1]}f(x)g(x)w(x) $, and that $\ell^\mst = (\ell_0,\ell_1, \ldots, \ell_\eta) = (\mcl(1),\mcl(x),\ldots, \mcl(x^\eta))$ is the vector corresponding to the linear functional $\mcl$, written with respect to the monomial basis. But the contraction $\ell^\mst  (H^{(\eta)})^{-1} \ell$ is of course independent of which basis of the space of degree-at-most-$\eta$-polynomials-on-$[0,1]$ we used to write down the matrices; using the basis given by the shifted Jacobi polynomials (with respect to which the (inverse) Gram matrix is diagonal) immediately yields $\ell^\mst  (H^{(\eta)})^{-1} \ell=\sum_{j=0}^\eta {\calL(J_j)^2}/{\|J_j\|^2}$, so that we are done.
\end{proof}
\begin{restatable}{lem}{lemjacobi}\label{lem:jacobi} 
The action of $\mcl$ on the shifted Jacobi polynomial $J_j$ satisfies
\begin{equation}
\calL(J_j)=\binom{n+1}{\eta}\frac{N-(\eta-j)_{j+1}/(N+1)_j}{M+j+1}\,,
\end{equation}
where $(a)_{j+1}=a(a+1)\cdots(a+j)$ is the Pochhammer symbol, and we adopt the convention  that  $(0)_{\eta+1}=0$.
\end{restatable}
\begin{proof}
The standard series expansion of the hypergeometric function ${}_2F_1$ gives
\begin{equation}\label{eq:jseries}
    J_j(x) = (j+1){}_2F_1(-j, j+M+2;2;x) = (j+1)\sum_{m=0}^j (-1)^m \binom{j}{m}\frac{(j+M+2)_m}{(2)_m}x^m \,,
\end{equation}
so that
\begin{align}
    \calL(J_j) &= \binom{n+1}{\eta}(j+1)\sum_{m=0}^j \frac{(-j)_m(j+M+2)_m}{(2)_m(N+1)_m}\\
     &= \binom{n+1}{\eta}(j+1)\sum_{m=0}^j \frac{(-j)_m(j+M+2)_m}{m!(N+1)_m}\frac{1}{m+1}\\
     &= \binom{n+1}{\eta}(j+1)\sum_{m=0}^j \frac{(-j)_m(j+M+2)_m}{m!(N+1)_m}\int_{z\in [0,1]}z^m\\
     &= \binom{n+1}{\eta}(j+1)\int_{z\in [0,1]}{}_2F_1(-j,j+M+2;N+1;z)\label{eq:fint}
\end{align}
Now, when $N+1>j+M+2$ (equivalently, $\eta>j$) we can insert Euler's beta integral form of ${}_2F_1$, giving
\begin{align}
    \calL(J_j) &= \binom{n+1}{\eta}(j+1)\int_{z\in [0,1]}{}_2F_1(-j,j+M+2;N+1;z)\\
    &= \binom{n+1}{\eta}(j+1)B(j+M+2, N-j-M-1)^{-1}\int_{z\in [0,1]} \int_{s\in [0,1]}  s^{j+M+1}(1-s)^{N-j-M-2}(1-zs)^j\\
    &= \binom{n+1}{\eta}B(j+M+2, N-j-M-1)^{-1}\int_{s\in [0,1]}  s^{j+M}(1-s)^{N-j-M-2}(1-(1-s)^{j+1})\\
    &= \binom{n+1}{\eta}B(j+M+2, N-j-M-1)^{-1}\left[\int_{s\in [0,1]}  s^{j+M}(1-s)^{N-j-M-2}-\int_{s\in [0,1]}  s^{j+M}(1-s)^{N-M-1} \right]\\
    &= \binom{n+1}{\eta}B(j+M+2, N-j-M-1)^{-1}\left[B(j+M+1, N-j-M-1)-B(j+M+1, N-M) \right]\\
    &= \binom{n+1}{\eta}\frac{\Gamma(N+1)}{\Gamma(j+M+2)\Gamma(N-j-M-1)}\left[\frac{\Gamma(j+M+1)\Gamma(N-j-M-1)}{\Gamma(N)}-\frac{\Gamma(j+M+1)\Gamma(N-M)}{\Gamma(N+j+1)}  \right]\\
    &= \binom{n+1}{\eta}\left(\frac{N}{j+M+1}-\frac{(N-j-M-1)_{j+1}}{(j+M+1)(N+1)_j}\right) 
\end{align}
Recalling the definitions $N=n-\eta+1=M+\eta+1$ concludes the proof when $\eta>j$. When $\eta=j$ we have $j+M+2=N+1$, so that the above integral form of the hypergeometric function is not valid. It is however easy to deal with this special case directly; indeed we have 
\begin{equation}
    {}_2F_1(-\eta,N+1;N+1;z) =\sum_{m=0}^\eta (-1)^m\binom{\eta}{m}\frac{(N+1)_m}{(N+1)_m}z^m =\sum_{m=0}^\eta \binom{\eta}{m} (-z)^m = (1-z)^\eta\,,
\end{equation}
so that from Eq.~\eqref{eq:fint} we have
\begin{align}
    \mcl(J_\eta)&=\binom{n+1}{\eta}(\eta+1)\int_{z\in [0,1]}(1-z)^\eta=\binom{n+1}{\eta}\,,
\end{align}
as in the claim of the lemma.
\end{proof}

At this point we are ready to give an expression for $\mu_0$ as a simple finite sum in one variable (recall that we let $M=n-2\eta,\ N=n-\eta+1=M+\eta+1$):
\begin{restatable}{lem}{lemmusum}\label{lem:musum} 
We have
\begin{equation}
  \mu_0=\frac{\eta+1}{(n+2)N}\sum_{j=0}^\eta\frac{M+2j+2}{(j+1)(M+j+1)}\left(N-\frac{(\eta-j)_{j+1}}{(N+1)_j}\right)^2\,.
\label{eq:singlesum}
\end{equation}
\end{restatable}
\begin{proof}
The linear algebraic fact that we need here is that the ratio of the determinant of the $\eta$\textsuperscript{th} Gram matrix to that of the $(\eta-1)$\textsuperscript{th} is exactly given by the squared norm of the corresponding monic orthogonal polynomial of degree $\eta$~\cite{lalley2013orthogonal}. Our shifted Jacobi polynomials are not monic; indeed Eq.~\eqref{eq:jseries} shows that  the leading coefficient is $\kappa_j=(-1)^j(j+M+2)_j/j!$.
Eq.~\eqref{eq:jnorm} and Lemmas~\ref{lem:mulin} and~\ref{lem:jacobi} then combine to yield
\begin{align}
    \mu_0 &=\frac{ \|J_\eta\|^2}{\kappa_\eta^2} \sum_{j=0}^\eta\frac{\calL(J_j)^2}{\|J_j\|^2}\\
    &=\frac{(\eta+1)(\eta!)^2}{(M+2\eta+2)(M+\eta+1)(\eta+M+2)_\eta^2} \sum_{j=0}^\eta \binom{n+1}{\eta}^2\left(\frac{N-(\eta-j)_{j+1}/(N+1)_j}{M+j+1}\right)^2\frac{(M+2j+2)(M+j+1)}{j+1} \\
    &=\frac{(\eta+1)(\eta!)^2}{(n+2)(n-\eta+1)(n-\eta+2)_\eta^2}\binom{n+1}{\eta}^2 \sum_{j=0}^\eta \frac{M+2j+2}{(j+1)(M+j+1)}\left(N-\frac{(\eta-j)_{j+1}}{(N+1)_j}\right)^2\\
    &=\frac{(\eta+1)((n+1)!)^2}{(n+2)N((n+1-\eta)!^2)(n-\eta+2)_\eta^2} \sum_{j=0}^\eta \frac{M+2j+2}{(j+1)(M+j+1)}\left(N-\frac{(\eta-j)_{j+1}}{(N+1)_j}\right)^2\\
    &=\frac{\eta+1}{(n+2)N} \sum_{j=0}^\eta \frac{M+2j+2}{(j+1)(M+j+1)}\left(N-\frac{(\eta-j)_{j+1}}{(N+1)_j}\right)^2
\end{align}
where we have used $M=n-2\eta$ and $N=n-\eta+1$.
\end{proof}
\noindent
We are now ready to prove Theorem~\ref{thm:1} (restated for convenience):
\thmone*
\begin{proof}
As established in the main text, 
$\Var [\widehat\Pi_\phi] \le \|\Gamma_R\|_\infty$, with $\Gamma_R$ as defined in Eq.~\eqref{eq:gammar}. The spectral norm of (the positive semidefinite) $\Gamma_R$ is given by its largest eigenvalue, which by Lemmas~\ref{lem:muzero} and~\ref{lem:musum} is given by 
\begin{equation}\label{eq:muthm1}
\mu_0=\frac{\eta+1}{(n+2)N}
  \sum_{j=0}^\eta
  \frac{M+2j+2}{(j+1)(M+j+1)}
  \left(N-\frac{(\eta-j)_{j+1}}{(N+1)_j}\right)^2\,,
\end{equation}
where $M=n-2\eta $ and $ N=n-\eta+1$. 
We now turn to bounding the various terms in Eq.~\eqref{eq:muthm1}, which can be done in an elementary fashion. Let us define $r_j= (\eta-j)_{j+1}/(N+1)_j$. We then calculate
\begin{equation}
r_j=\frac{(\eta-j)(\eta-j+1)\cdots\eta}{(N+1)(N+2)\cdots(N+j)} =(j+1)\frac{\binom\eta{j+1}}{\binom{N+j}{j}} =\eta\frac{\binom{\eta-1}{j}}{\binom{N+j}{j}}\le \eta< N\,,
\end{equation}
where we have used that $\eta\le n/2$ implies $N>\eta$. It follows that $(N-r_j)^2\le N^2$, the substitution of which into Eq.~\eqref{eq:muthm1} yields
\begin{equation}
    \mu_0\le (\eta+1) 
  \sum_{j=0}^\eta \frac{1}{j+1}
  \frac{N(M+2j+2)}{(n+2)(M+j+1)}\,.
\end{equation}
Finally, inserting $N=M+\eta+1$ yields $(n+2)(M+j+1)-N(M+2j+2) = M(\eta-j)\ge 0$, whence $N(M+2j+2)/((n+2)(M+j+1))\le 1$, so that 
\begin{equation}
    \mu_0\le (\eta+1) 
  \sum_{j=0}^\eta \frac{1}{j+1} = (\eta+1)H_{\eta+1}\,,
\end{equation}
establishing the theorem. 
\end{proof}
\noindent
We now turn to the (significantly less involved) proof of Theorem~\ref{thm:tight}:
\thmtight*
\begin{proof}
We without loss of generality pick a basis of $\mbc^n$ so that $\Pi_\phi = \ketbra{1,2,\ldots \eta} = \Pi_R$. Now, 
the assumption that the input state $\rho$ \textit{equals} the Slater determinant we are trying to  estimate simplifies things enormously. First, from $\Var [\widehat\Pi_R] = \expect [\widehat\Pi_R^2] - \expect[\widehat\Pi_R]^2$, we have $\Var [\widehat\Pi_R]=\expect [\widehat\Pi_R^2] -1$; second, from the previously seen fact that $  \Gamma_R|_{\mch_s}=\mu_s \id_{\mch_s}$, where $\mch_s={\rm span}\,\{|S\rangle:d(S,R)=s\}$ is the Johnson shell at distance $s$ from $R$, we have $ \expect [\widehat\Pi_R^2] = \Tr[\Gamma_R \rho]= \Tr[\Gamma_R \ketbra{R}]=\mu_0$. Indeed, $R$ is at a Johnson distance of zero from $R$. Our task is therefore just to  asymptotically analyse $\mu_0$ under the assumption $n=c\eta$. Starting from the expression Eq.~\eqref{eq:singlesum}, and defining for convenience $M=n-2\eta,\ N=n-\eta+1, p_j = (\eta-j)_{j+1}/ (N(N+1)_j)$ and $b_j=N(M+2j+2)/((n+2)(M+j+1))$, we have 
\begin{align}
  \frac{\mu_0}{\eta+1}&=\frac{1}{(n+2)N}\sum_{j=0}^\eta\frac{M+2j+2}{(j+1)(M+j+1)}\left(N-\frac{(\eta-j)_{j+1}}{(N+1)_j}\right)^2\\
  &=\sum_{j=0}^\eta\frac{b_j}{j+1 }\left(1-p_j \right)^2\,.
\end{align}
One obtains the upper bound of Theorem~\ref{thm:1} by the bound $b_j (1-p_j)^2\le 1$. 
The technical observation we need is that most of the mass of the harmonic sum lives below $\delta\eta$ for any fixed $\d>0$; indeed from $\log(\d\eta) = \log\eta + \log\d$ we have $\sum_{j=0}^{\delta\eta}1/(j+1)=\log\eta + \mco_\d(1)$, where for some error term $\mco_\d(1)$ denotes a ``constant'' which may depend on $\d$. Now, with $n=c\eta$ we have $M=(c-2)\eta,\ N\sim (c-1)\eta$, so that 
\begin{equation}
    b_j = \frac{(c-1+1/\eta)((c-2)+2j/\eta+2/\eta)}{(c+2/\eta)((c-2)+j/\eta+1/\eta)} = \frac{c-1}{c} + \mco_c(\d) + o(1)\,,
\end{equation}
where we assume $j/\eta\le \d$, and the $o(1)$ term signifies an error that will vanish as $\eta\to\infty$. Meanwhile, the numerator of $p_j$ contains $j+1$ factors (each at most $\eta$) while its denominator contains $j+1$ factors (each at least $N$). As $\eta/N\to 1/(c-1)<1$ there is evidently some $q<1$ such that $p_j\le q^{j+1}$, from which we conclude that $(1-p_j)^2=1+\mco(p_j)$ and that 
\begin{equation}
    \sum_{j=0}^\eta \frac{p_j}{j+1}\le\sum_{j=0}^\eta \frac{q^{j+1}}{j+1}\le\sum_{j=0}^\infty \frac{q^{j+1}}{j+1} = \mco_c(1)\,,
\end{equation}
as $q$ depended on $c$. Putting it all together we have
\begin{equation}
    \frac{\mu_0}{\eta+1}= \sum_{j=0}^{\d\eta}\frac{b_j}{j+1 }\left(1-p_j \right)^2+  \sum_{j=\d\eta}^{\eta}\frac{b_j}{j+1 }\left(1-p_j \right)^2 = \left(\frac{c-1}{c} + \mco_c(\d)  + o(1)\right)H_{\eta+1} +\mco_c(1)+  \mco_\d(1)\,.
\end{equation}
Dividing by $H_{\eta+1}$, taking the limit $\eta\to\infty$, and then taking the limit $\d\to 0$ establishes the claim of the theorem. 
    
\end{proof}

\noindent
We now turn to the proof of Theorem~\ref{thm:comp}:
\thmcomp*
\begin{proof}
Let us write the target state $\ket\phi = \phi_1\wedge\cdots \wedge \phi_\eta$ as an $n\times \eta$ matrix, $\Phi = [\phi_1,\ldots \phi_\eta]$, extended to a unitary $W = [\Phi, \Phi_\perp]\in U(n)$. The introduction of $W$ is simply for algebraic convenience in the proof; we will never have to actually calculate $\Phi_\perp$, which will in fact disappear from the final answer. Now, with $v_Z$ such that $v_ZR=Z$,
\begin{align}
\widehat\Pi_\phi(u,Z)&=\Tr\left[\Pi_\phi U_\eta(u)^\dagger \mcm^{-1}(\Pi_Z)U_\eta(u)\right]\\
  &=\Tr\left[U_\eta(W)\Pi_R U_\eta(W)^\dagger U_\eta(u)^\dagger \mcm^{-1}(U_\eta(v_Z)\Pi_RU_\eta(v_Z\ad))U_\eta(u)\right]\\
  &=\Tr\left[\Pi_R U_\eta(v_Z\ad u W)^\dagger \mcm^{-1}(\Pi_R)U_\eta(v_Z\ad u W)\right]\\
  &=g_{\eta}\left([(v_Z\ad u W)_{R,R}]\ad [(v_Z\ad u W)_{R,R}] \right)\,,
\end{align}
where we have used Lemma~\ref{lem:hlem}. The defining properties of $v_Z$ and $W$ then give:
\begin{equation}
g_{\eta}\left([(v_Z\ad u W)_{R,R}]\ad [(v_Z\ad u W)_{R,R}] \right)
  =g_{\eta}\left([(  u W)_{Z,R}]\ad [( u W)_{Z,R}] \right)
  =g_{\eta}\left([  u_{Z,:} W_{:,R}]\ad [ u_{Z,:} W_{:,R}] \right)
  =g_{\eta}\left([  u_{Z,:} \Phi_{:,R}]\ad [ u_{Z,:} \Phi_{:,R}] \right)\,,
\end{equation}
so that, with $A=u_{Z,:} \Phi_{:,R}\in\mbc^{\eta\times\eta}$, we have $\widehat\Pi_\phi(u,Z)=g_\eta(A\ad A)$. Calculating the matrix product which defines $A$ takes (using, say, naive matrix multiplication algorithms) time $\mco(n\eta^2)$; at the price of rather large constants one could improve this slightly using fast matrix multiplication algorithms. Calculating the product $A\ad A$ is $\mco(\eta^3)$; as $\eta\le n/2$ this is no worse than $\mco(n\eta^2)$. Now, what is the cost of evaluating $g_\eta(A\ad A)$?  By Lemma~\ref{lem:hlem} we can evaluate $g_\eta(A\ad A)$ performing a weighted sum (with easily computable weights) over the first $\eta$ (well, $\eta+1$ if we count the trivial  $e_0$) elementary symmetric polynomials in the eigenvalues of $A\ad A$ (themselves acquirable in time $\mco(\eta^3)$). The first $\eta$  elementary symmetric polynomials can be calculated by the following simple recursive algorithm: First, initialise $e_0=1$, $e_r=0$ for all $r>0$. Then, for each eigenvalue $\lm_i$, update to polynomials according to $e_r\leftarrow e_r + \lm_i e_{r-1}$, where we update ``downwards'',  i.e. in the order $r=i,i-1,\ldots 1$. This takes time $\mco(\eta^2)$, and so does not affect the  asymptotic runtime of the full  algorithm, which remains at $\mco(n\eta^2)$.
\end{proof}

\section{Zonal spherical functions and shell eigenvalues}\label{sec:zsf}
In this appendix we prove Lemma~\ref{lem:muzero} by employing some ideas from the theory of \textit{zonal spherical functions}, the pertinent details of which we begin by briefly reviewing. Our starting point is the fact that $(G,K)=( U(n), U(\eta)\times  U(n-\eta))$ is a \textit{compact Gelfand pair}. For our purposes the relevant consequences of this begin with the fact that the Hilbert space $L^2(G/K)$ decomposes multiplicity-freely under the (left) action of the group, 
\begin{equation}\label{eq:l2}
L^2(G/K)= \bigoplus_{\lambda\in\Omega}\mch_\lambda\,,    
\end{equation}
where each irrep $\mch_\lm$ possesses a unique one-dimensional subspace of $K$-fixed vectors. From each such line we choose a unit vector $v_\lambda$.  The corresponding zonal spherical function is then defined to be
\begin{equation}
  \phi_\lambda(gK)=\braket{v_\lambda}{\pi_\lambda(g)v_\lambda}\,,
\end{equation}
where $\pi_\lambda$ is the representation of $G$ on $\mch_\lambda$; for our choice of $G$ and $K$, the spherical functions can be taken to be real-valued~\cite{remling2015convolution}. The Schur-orthogonality relation implies that the spherical functions are orthogonal in the sense that
\begin{equation}\label{eq:ortho}
\int_{x\sim G/K} \phi_\lambda(g^{-1}x)\phi_\mu(x) =\frac{\delta_{\lambda\mu}}{d_\lambda}\phi_\lambda(g)\,.
\end{equation}

Now, let us call a  $K$-bi-invariant function $f$ on $G$ \textit{positive-semidefinite} if, for every finite list $g_1,\ldots,g_m\in G$, the matrix
$\bigl[f(g_i^{-1}g_j)\bigr]_{i,j=1}^m$
is positive-semidefinite. For example, let us see that the spherical functions (lifted in the obvious way to functions on $G$) are  positive-semidefinite. Firstly, they are $K$-bi-invariant: 
\begin{align}
    \phi_\lm(k_1gk_2) &= \braket{v_\lambda}{\pi_\lambda(k_1gk_2)v_\lambda}\\
    &= \braket{v_\lambda}{\pi_\lambda(k_1)\pi_\lm(g)\pi_\lm(k_2)v_\lambda}\\
    &= \braket{\pi_\lambda(k_1^{-1})v_\lambda}{\pi_\lm(g)\pi_\lm(k_2)v_\lambda}\\
    &= \braket{v_\lambda}{\pi_\lm(g) v_\lambda}\\
    &= \phi_\lm(g)
\end{align}
Secondly, for arbitrary $\{c_i\}_i\subset\mbc$ and $\{g_i\}_i\subset G$, we have
\begin{equation}
    \sum_{i,j}\overline{c_i}c_j\phi_\lm(g_i^{-1}g_j) = \sum_{i,j}\overline{c_i}c_j\braket{v_\lambda}{\pi_\lm(g_i^{-1}g_j) v_\lambda} = \Big\langle \sum_{i}c_i\pi_\lm(g_i)v_\lambda \Big\vert \sum_{j}c_j\pi_\lm(g_j) v_\lambda \Big\rangle \ge 0\,.
\end{equation}
More generally, any $f=\sum_\lm a_\lm \phi_\lm$ for non-negative coefficients $a_\lm$ is positive-semidefinite; conversely, the Bochner–Godement theorem states that a positive-semidefinite $K$-bi-invariant function has such an expansion, with nonnegative spherical coefficients $a_\lm$. Our next task is to prove:
\begin{restatable}{lem}{lemposdef}\label{lem:posdef}
The functions $p_0$ and $F^2$ are positive-semidefinite.
\end{restatable} 
\begin{proof}
We begin by showing that $p_0$ is positive-semidefinite: indeed, $p_0(gK)=\lvert\langle R|U_\eta(g)|R\rangle\rvert^2
  =\langle w,(U_\eta(g)\otimes\overline{U_\eta(g)})w\rangle $, where $
  w=|R\rangle\otimes\overline{|R\rangle}$, so that $p_0$
is a diagonal matrix coefficient of a unitary representation, and therefore positive-semidefinite. From the Bochner–Godement theorem  we therefore have an expansion
\begin{equation}\label{eq:p0sum}
  p_0=\sum_{\lambda\in H}b_\lambda\phi_\lambda
\end{equation}
for some subset $H$ of the irreps of $G$, with all $b_\lambda>0$. To show that $F^2$ is positive-semidefinite we find it  helpful to introduce the map
\begin{equation}
  \Phi:\operatorname{End}(\mch_\eta)\longrightarrow L^2(G/K),
  \qquad (\Phi A)(x)=\Tr(\Pi_xA)\,,
\end{equation}
where, for any $x\in gK$, we define $  \Pi_x=U_\eta(g)\Pi_RU_\eta(g)^\dagger$.
The adjoint of $\Phi$ is given by 
\begin{equation}
\Phi^*:L^2(G/K)\longrightarrow \operatorname{End}(\mch_\eta),
  \qquad\Phi^*f=\int_{x\sim G/K}f(x)\Pi_x\,.
\end{equation}
By the cyclicity of the trace, we have from Eq.~\eqref{eq:h} that
\begin{equation}
F(w)=\Tr\left[\Pi_RU_\eta(w)^\dagger \mcm^{-1}(\Pi_R)U_\eta(w)\right]= \Tr\left[\Pi_w  \mcm^{-1}(\Pi_R)\right]=   (\Phi \mcm^{-1}(\Pi_R))(w)\,,
\end{equation}
so that $F=\Phi (\mcm^{-1}(\Pi_R))$. As we have also that $\mcm=\binom{n}{\eta}\Phi^*\Phi$ and $p_0=\Phi (\Pi_R)$, defining $T=\binom{n}{\eta}\Phi\Phi^*$ it follows that
\begin{equation}\label{eq:htp}
    F = \Phi (\mcm^{-1}(\Pi_R))= \binom{n}{\eta}^{-1} \Phi ((\Phi^*\Phi)^{-1}(\Pi_R))= \binom{n}{\eta}^{-1} (\Phi\Phi^*)^{-1}\Phi (\Pi_R) = T^{-1}\Phi (\Pi_R) = T^{-1}p_0\,,
\end{equation}
where the ``inverses'' of the maps are taken to be their corresponding Moore-Penrose pseudo-inverses. 
Now, $T$  is $G$-equivariant: indeed, if $L_a$ denotes left translation, $(L_af)(y)=f(a^{-1}y)$, then unraveling the definitions quickly yields $T(L_af)=L_a(Tf)$. Combined with the multiplicity-freeness of $L^2(G/K)$, Schur's lemma forces $  T|_{\mch_\lambda}=a_\lambda \id_{\mch_\lm}$ for all $\lm$, and in particular $  T\phi_\lambda=a_\lambda\phi_\lambda$. We can calculate the eigenvalues directly; indeed
\begin{align}
    (T\phi_\mu)(y) &= \binom{n}{\eta}\int_{x\sim G/K}\phi_\mu(x)\Tr[\Pi_x\Pi_y]\\
    &=\binom{n}{\eta}\int_{x\sim G/K}\phi_\mu(x)p_0(y^{-1}x)\\
    &=\binom{n}{\eta}\int_{x\sim G/K}\phi_\mu(x)\sum_{\lambda\in H}b_\lambda\phi_\lambda(y^{-1}x)\\
    &=\binom{n}{\eta}\frac{b_\mu}{d_\mu}\phi_\mu(y)\,,
\end{align}
where we have used Eqs.~\eqref{eq:ortho} and~\eqref{eq:p0sum}. So $a_\mu=\binom{n}{\eta}b_\mu/d_\mu\ge 0$. Eq.~\eqref{eq:htp} then implies that
\begin{equation}
    F=T^{-1}p_0 = \sum_{\lambda\in H}\frac{b_\lambda}{a_\lm}\phi_\lambda =\binom{n}{\eta}^{-1} \sum_{\lambda\in H} d_\lm\phi_\lambda 
\end{equation}
is positive-semidefinite. The  Schur product theorem then implies that $F^2$ is  in turn positive-semidefinite.
\end{proof}
In our next lemma we relate the spherical expansion of $p_s$ to that of $p_0$. Abusing notation slightly by identifying a point $x\in G/K$ with the subspace of $\mbc^n$ spanned by its first $\eta$ columns (that is, a point of the \textit{Grassmannian} ${\rm Gr}(\eta,n)$) we have:
\begin{restatable}{lem}{lempssum}\label{lem:psdef}
Let $x_s\in G/K$ be any coordinate $\eta$-plane at Johnson distance $s$ from $R$. Then we have $ p_s=\sum_\lambda b_\lambda \phi_\lambda(x_s) \phi_\lambda$, where the $b_\lm$ are the spherical coefficients of the expansion of $p_0$. 
\end{restatable} 
\begin{proof}
We begin by observing that for $y\in G/K$ we have
\begin{equation}
    p_s(y) = \frac{1}{N_s}\Tr[\Pi_y P_s] = \Tr[\Pi_y \int_{k\sim K}\Pi_{kx_s}] = \int_{k\sim K}p_0(y^{-1}kx_s)\,,
\end{equation}
where in the second equality we have used Schur's lemma; indeed notice that $\int_{k\sim K}\Pi_{kx_s}$ commutes with the action of $K$, while having support only on the $K$-irrep $\mch_s$ (c.f. the proof of Lemma~\ref{lem:gamma}). Inserting the expansion Eq.~\eqref{eq:p0sum} then yields
\begin{align}
    p_s(y) &= \int_{k\sim K}\sum_{\lambda\in H}b_\lambda\phi_\lambda(y^{-1}kx_s)\\
    &= \int_{k\sim K}\sum_{\lambda\in H}b_\lambda   \braket{v_\lambda}{\pi_\lambda(y^{-1}kx_s)v_\lambda}  \\
    &= \sum_{\lambda\in H}b_\lambda  \Big\langle v_\lambda \Big\vert \pi_\lambda(y^{-1})\left(\int_{k\sim K} \pi_\lm(k)\right)\pi_\lm(x_s)v_\lambda\Big\rangle  \\
    &= \sum_{\lambda\in H}b_\lambda  \braket{v_\lambda}{\pi_\lambda(y^{-1}) \ketbra{v_\lm}{v_\lm} \pi_\lm(x_s)v_\lambda}  \\
    &= \sum_{\lambda\in H}b_\lambda \phi_\lm(x_s)\phi_\lm (y^{-1})  
\end{align}
where we have used that $\int_{k\sim K} \pi_\lm(k)$ projects onto the $K$-fixed subspace of $\mch_\lm$, which is spanned exactly by $v_\lm$. The result of Lemma~\ref{lem:psinv} now immediately concludes the proof. 
\end{proof}
We have now finally assembled the ingredients necessary to prove:
\lemmuzero*
\begin{proof}
    By Lemmas~\ref{lem:posdef} and~\ref{lem:psdef} we have spherical expansions
    \begin{align}
        p_0&=\sum_\lambda b_\lambda  \phi_\lambda\\
        p_s&=\sum_\lambda b_\lambda \phi_\lambda(x_s) \phi_\lambda\\
        F^2&=\sum_\lambda c_\lambda \phi_\lambda
    \end{align}
    where $b_\lm,\,c_\lm\ge 0$ for all $\lm$. By Lemma~\ref{lem:muint} we then find
    \begin{align}
  \mu_s&=\binom n\eta\int_{x\sim G/K}F(x)^2p_s(x)\\
  &=\binom n\eta\int_{x\sim G/K}\sum_{\lm,\mu}b_\lambda c_\mu \phi_\lambda(x_s) \phi_\lambda(x)\phi_\mu(x)\\
  &=\binom n\eta \sum_{\lm}\frac{b_\lambda c_\lm }{d_\lm}  \phi_\lambda(x_s) \,,
    \end{align}
    where we have used the spherical function orthogonality relation Eq.~\eqref{eq:ortho}. Now, as the $ \phi_\lambda(x_s)$ are matrix elements of  unitary representations, their magnitudes are bounded by one, so that we have
    \begin{equation}
        \mu_s = \binom n\eta \sum_{\lm}\frac{b_\lambda c_\lm }{d_\lm}  \phi_\lambda(x_s)\le \binom n\eta \sum_{\lm}\frac{b_\lambda c_\lm }{d_\lm}  \lvert\phi_\lambda(x_s)\rvert\le \binom n\eta \sum_{\lm}\frac{b_\lambda c_\lm }{d_\lm}  = \mu_0\,,
    \end{equation}
    concluding the proof.
\end{proof}

\section{Estimating particle-preserving quadratic hamiltonians}\label{sec:ham}
Our first task in this appendix is to build towards a proof of Theorem~\ref{thm:ham}. As this will involve quite a lot of switching backwards and forwards between working in $\mbc^n$ and $\mch_\eta\cong\mbc^{\binom{n}{\eta}}$, we begin by giving an explicit example of what the various objects we will be considering look like in these two spaces. 
First, recall that a   point on the Grassmannian ${\rm Gr}(\eta,n)$ is equivalently thought of as an $\eta$-dimensional  subspace of $\C^n$, which may be specified by some  rank-$\eta$  projector $P$. 
If $f_1,\ldots,f_\eta$ is any orthonormal basis of the span of $P$, then the corresponding Slater determinant state is $\ket{P}=f_1\wedge\cdots\wedge f_\eta\in\Lambda^\eta\C^n$, 
and the corresponding many-body (rank-one) projector is $\Pi_P=\ketbra{P}{P}\in{\rm End}\,(\mch_\eta)$ (note that this definition is independent of the orthonormal basis of the span of $P$, since changing the basis by a unitary simply multiplies the wedge by a phase; indeed if $u$ preserves the span of $P$ then $U_\eta(u)\ket P = \det (u_{P,P}) \ket P$). Let us give an explicit example with  $n=4,\eta=2$. Taking (for no particular reason) $f_1=(e_1+e_3)/\sqrt2,\
  f_2=( e_2+e_4)/\sqrt2 $ yields
  \begin{equation}
      P=\ketbra{f_1}+\ketbra{f_2}
  =\frac12\begin{pmatrix}
  1&0&1&0\\
  0&1&0&1\\
  1&0&1&0\\
  0&1&0&1
  \end{pmatrix}\,;
  \end{equation}
  meanwhile, the associated Slater determinant is $\ket{ P}=f_1\wedge f_2
  =\frac12\big(\ket{12}+\ket{14}-\ket{23}+\ket{34}\big)$. With respect, then, to the   two-particle basis $ \ket{12},\ket{13},\ket{14},\ket{23},\ket{24},\ket{34},$
the  projector $\Pi_P=\ketbra{P}$ is the rank-one matrix
\begin{equation}
  \Pi_P=\frac14
  \begin{pmatrix}
  1&0&1&-1&0&1\\
  0&0&0&0&0&0\\
  1&0&1&-1&0&1\\
  -1&0&-1&1&0&-1\\
  0&0&0&0&0&0\\
  1&0&1&-1&0&1
  \end{pmatrix}.
\end{equation}
In the sequel, we will need the following identity:
\begin{equation}
  \Tr_{\mch_\eta}\bigl(d{U_\eta}(h)\Pi_P\bigr)=\Tr_{\C^n}(Ph)\,;
  \label{eq:trinv}
\end{equation}
indeed, unraveling the definitions we have
\begin{equation}
    \Tr_{\mch_\eta}\big(d{U_\eta}(h)\Pi_P\big) = \braket{P}{d{U_\eta}(h)|P}= \bra{f_1\wedge\cdots\wedge f_\eta}\sum_{r=1}^\eta\ket{f_{1}\wedge\cdots\wedge(hf_{r})\wedge\cdots\wedge f_{\eta}} = \sum_{r=1}^\eta  \braket{f_r}{hf_r}=\Tr_{\C^n}(Ph)\,.
\end{equation}
Now, let us consider the action of the shadow channel on $d{U_\eta}(h)$. Happily,  this is quite a lot simpler than the Slater determinant case:
\begin{restatable}{lem}{lemmquad}\label{lem:mquad} 
For every traceless Hermitian $h\in{\rm End}\,(\C^n)$ we have
\begin{equation}
  \mcm(d{U_\eta}(h))=\frac1{n+1}d{U_\eta}(h)\,.
  \label{eq:mquad}
\end{equation}
\end{restatable}
\begin{proof}
The important representation-theoretic input here is that $dU_\eta$ intertwines the one-body conjugation
action with the many-body conjugation action; that is, $  d{U_\eta}(uhu\ad)=U_\eta(u)d{U_\eta}(h)U_\eta(u)\ad$. So, $d{U_\eta}:\mfu(n)\hookrightarrow{\rm End}\,(\mch_\eta)$ embeds the traceless one-body hamiltonians into a copy of the adjoint module of $ U(n)$ in ${\rm End}\,(\mch_\eta)$. The multiplicity-freeness of the decomposition of  ${\rm End}\,(\mch_\eta)$ into $ U(n)$-modules combined with the (previously discussed) equivariance of $\mcm$ and Schur's lemma then implies that we must have $  \mcm(d{U_\eta}(h))=\lambda d{U_\eta}(h)$ for some scalar $\lm$, the computation of which is our remaining job~\cite{fulton1991representation}. We calculate:
\begin{align}
    \Tr_{\mch_\eta} [d{U_\eta}(h)\mcm(d{U_\eta}(h))] &= \Tr_{\mch_\eta} [d{U_\eta}(h)\int_{u\sim U(n)} \sum_{Z\in\mcz_\eta} \Tr_{\mch_\eta}[d{U_\eta}(h) U_\eta(u)\ad \Pi_ZU_\eta(u)]U_\eta(u)\ad \Pi_Z U_\eta(u)]\\
     &= \binom{n}{\eta} \int_{u\sim U(n)}   \Tr_{\mch_\eta}[d{U_\eta}(h) U_\eta(u)\ad \Pi_RU_\eta(u)]\Tr_{\mch_\eta} [d{U_\eta}(h)U_\eta(u)\ad \Pi_R U_\eta(u) ]\\
     &= \binom{n}{\eta} \int_{P\sim {\rm Gr}(\eta, n)}   \Tr_{\mch_\eta}[d{U_\eta}(h)   \Pi_P ]^2\\
     &= \binom{n}{\eta} \int_{P\sim {\rm Gr}(\eta, n)}   \Tr_{\mbc^n}[ h P ]^2
\end{align}
where we have used Eq.~\eqref{eq:trinv} and converted to an integral over the Grassmannian ${\rm Gr}(\eta, n)$, parameterised by subspaces $P=UQ_RU\ad$ of $\mbc^n$, where $  Q_R=\begin{pmatrix}I_\eta&0\\0&0\end{pmatrix}$. The final integral is that of a Haar-random rank-$\eta$ projector, which may be evaluated via elementary Weingarten calculus~\cite{mele2024introduction}; one concludes
\begin{equation}
    \Tr_{\mch_\eta} [d{U_\eta}(h)\mcm(d{U_\eta}(h))] = \binom{n}{\eta} \frac{\eta(n-\eta)}{n(n^2-1)}\Tr[h^2]\,.
\end{equation}
On the other hand, we know that $d{U_\eta}(h)$ is an eigenvector of $\mcm$ (of as-yet unknown eigenvalue $\lm$), so that (taking wlog $h={\rm diag}\,(x_1,\ldots,x_n)$)   for any $Z\in\mcz_\eta$ we have $  d{U_\eta}(h)\ket Z=\left(\sum_{i\in Z}x_i\right)\ket Z,
$ whence
\begin{align}
\Tr_{\mch_\eta} [d{U_\eta}(h)\mcm(d{U_\eta}(h))] &= \lm\Tr_{\mch_\eta}(d{U_\eta}(h)^2)\\
&=\lm\sum_{Z\in\mcz_\eta}\left(\sum_{i\in Z}x_i\right)^2         \\                      
  &=\lm\binom{n-1}{\eta-1}\sum_i x_i^2
  +2\binom{n-2}{\eta-2}\sum_{i<j}x_ix_j\\
  &=\lm
  \left[\binom{n-1}{\eta-1}-\binom{n-2}{\eta-2}\right]\Tr(h^2)\\
  &=\lm
  \binom{n-2}{\eta-1}\Tr(h^2)        
\end{align}
Here we have used the tracelessness of $h$, as well as that each $x_i^2$ occurs in $\binom{n-1}{\eta-1}$ subsets, and each cross term $x_ix_j$ occurs in $\binom{n-2}{\eta-2}$ subsets.  It follows (well, assume $h$ is non-zero) that  we have
\begin{equation}
    \lm =  \binom{n}{\eta} \frac{\eta(n-\eta)}{n(n^2-1)}\binom{n-2}{\eta-1}^{-1} = \frac{1}{n+1}\,,
\end{equation}
which concludes the proof.   
\end{proof}
Lemma~\ref{lem:mquad} immediately yields an explicit expression for the single-shot estimator of a traceless particle-preserving quadratic operator $h_0$; indeed, we have
\begin{align}
  \widehat{o}_{h_0}(u,Z)  &=\Tr\left(d{U_\eta}(h_0)\mcm^{-1}\left(U_\eta(u)\ad\Pi_ZU_\eta(u)\right)\right)\\
    &=(n+1)\Tr\left(d{U_\eta}(h_0)U_\eta(u)\ad\Pi_ZU_\eta(u)\right)\\
  &=(n+1)\Tr\left(d{U_\eta}(uh_0u\ad)\Pi_Z\right)\\
  &=(n+1)\Tr[ uh_0u\ad Q_Z]\,,\label{eq:hest}
\end{align}
where $Q_Z = \sum_{i\in Z}\ketbra{i}$ and we have used Eq.~\eqref{eq:trinv}. The extension to traceful operators $h$ is immediate; indeed we have 
\begin{align}
    \widehat{o}_{\id_n}(u,Z)&=\Tr\left(d{U_\eta}(\id_n)\mcm^{-1}\left(U_\eta(u)\ad\Pi_ZU_\eta(u)\right)\right)\\
    &= \Tr\left( \eta\id_{\mch_\eta}\mcm^{-1}\left(U_\eta(u)\ad\Pi_ZU_\eta(u)\right)\right)\\
    &=\eta \Tr\left( \mcm^{-1}\left(\id_{\mch_\eta}\right)U_\eta(u)\ad\Pi_ZU_\eta(u)\right)\\
    &=\eta \Tr\left( U_\eta(u)\ad\Pi_ZU_\eta(u)\right)\\
    &=\eta
\end{align}
so that for a general particle-preserving
quadratic operator $h=\Tr[h]\id_n/n + h_0$ we have
\begin{equation}\label{eq:hestfull}
    \widehat{o}_{h}(u,Z)  =\eta\Tr[h]/n + (n+1)\Tr[ uh_0u\ad Q_Z]\,.
\end{equation}
Now, as in the case of estimating Slater determinants, our primary objective boils down to bounding the infinity norm of a certain second moment operator, given in this instance (for a target hamiltonian $h$) by
\begin{equation}
      \Gamma_h=\int_{u\sim U(n)}\sum_{Z\in\mcz_\eta}\widehat{o}_h(u,Z)^2U_\eta(u)\ad\Pi_ZU_\eta(u)\,.
  \label{eq:gammah}
\end{equation}
To get some technical preliminaries out of the way, we return to the consideration of some properties of the measure Eq.~\eqref{eq:mario}. In particular, we prove:
\begin{restatable}{lem}{lemjactwo}\label{lem:jactwo} 
Let $x_1,\ldots,x_\eta$ have density proportional to $\mu(x)=
  \Delta(x)^2\prod_{i=1}^\eta x_i(1-x_i)^{n-2\eta}$, and set $s_1=\sum_i x_i,$ and $s_2=\sum_i x_i^2$.
Then
\begin{align}
  \expect_{x\sim\mu}[s_1]&=\frac{\eta(\eta+1)}{n+1},\label{eq:es1}\\
  \expect_{x\sim\mu}[s_1^2]
  &=
  \frac{\eta(\eta+1)(n\eta^2+n\eta+n+\eta^2-\eta)}{n(n+1)(n+2)},\label{eq:es1sq}\\
  \expect_{x\sim\mu}[s_2]
  &=
  \frac{\eta(\eta+1)(2n\eta+n-\eta^2+\eta)}{n(n+1)(n+2)}.\label{eq:es2}
\end{align}
\end{restatable}
\begin{proof}
Our trick is going to be to obtain multiple equations in the unknown expectation values by integrating various functions against $\mu$. First, notice that
\begin{equation}\label{eq:partial}
    \partial_i \log \mu(x) =   2\sum_{j\ne i}\frac1{x_i-x_j}+\frac1{x_i}-\frac{n-2\eta}{1-x_i}.
\end{equation}
So, if we integrate a vector field $\varphi_i$ containing a factor of $x_i(1-x_i)$ for each $i$, the vanishing at the boundary yields
\begin{equation}\label{eq:intparts}
    0 = \int_{x\in[0,1]^\eta}\sum_i\partial_i\left[\varphi_i(x) \mu(x)\right]= \sum_i\int_{x\in[0,1]^\eta}\mu(x)\left[\partial_i\varphi_i(x) + \varphi_i (x)\partial_i\log\mu(x)\right] = \sum_i\expect_{x\sim\mu}\left[\partial_i\varphi_i(x) + \varphi_i (x)\partial_i\log\mu(x)\right]
\end{equation}
(note that the singularities at $x_i=x_j$ in Eq.~\eqref{eq:partial} are harmless due to the quadratic vanishing of the Vandermonde). Now, let us pick some vector fields and evaluate Eq.~\eqref{eq:intparts}. First we take $ \varphi_i=x_i(1-x_i)$, so that $ \sum_i\partial_i\varphi_i=\eta-2s_1$.  The Vandermonde term may be simplified by pairing the contribution from $(i,j)$ with that of $(j,i)$:
\begin{align}
2\sum_i x_i(1-x_i)\sum_{j\ne i}\frac1{x_i-x_j}  &=2\sum_{i<j}\left[\frac{x_i(1-x_i)}{x_i-x_j}+\frac{x_j(1-x_j)}{x_j-x_i}\right]\\
  &=2\sum_{i<j}(1-x_i-x_j)\\
  &=\eta(\eta-1)-2(\eta-1)s_1.
\end{align}
Multiplying the other two terms on the RHS of Eq.~\eqref{eq:partial} and summing over $i$ yields respectively $\eta-s_1$ and $-(n-2\eta)s_1$; adding up the various contributions then reveals
\begin{equation}
    0=\expect_{x\sim\mu}[\eta-2s_1 + \eta(\eta-1)-2(\eta-1)s_1+\eta-s_1+(2\eta-n)s_1]\,,
\end{equation}
whence
\begin{equation}\label{eq:s1}
  \expect_{x\sim\mu} s_1=  \frac{\eta(\eta+1)}{n+1}\,.
\end{equation}
Playing the same game with the vector fields $ \varphi_i=x_i^2(1-x_i)$ and $ \varphi_i=x_i(1-x_i)s_1$ yield respectively the equations
\begin{equation}
    0 = (2\eta+1)\expect_{x\sim\mu}[s_1] - \expect_{x\sim\mu} [s_1^2] - (n+1)\expect_{x\sim\mu} [s_2]
\end{equation}
and
\begin{equation}
    0 = (\eta^2+\eta+1)\expect_{x\sim\mu}[s_1] -(n+1) \expect_{x\sim\mu} [s_1^2] - \expect_{x\sim\mu} [s_2]\,.
\end{equation}
Substituting in the result of Eq.~\eqref{eq:s1} and solving simultaneously for the remaining two unknowns gives the claim of the lemma. 
\end{proof}

The following   technical lemma builds fairly directly on the result of Lemma~\ref{lem:jactwo}, combined with the (readily verified by elementary Weingarten calculus) identity
\begin{equation}
    \int_{V\sim U(d)} \left(\Tr [V {\rm diag}\,(\lambda_1,\ldots,\lambda_d) V\ad A]\right)^2 =   \frac{dr_1^2-r_2}{d(d^2-1)}(\Tr A)^2+\frac{dr_2-r_1^2}{d(d^2-1)}\Tr(A^2),
  \label{eq:fixedspec}
\end{equation}
where $  r_1=\sum_i \lm_i$ and $  r_2=\sum_i \lm_i^2$.
\begin{restatable}{lem}{lemjacthree}\label{lem:jacthree} 
Let $h={\rm diag}\,(\lambda_1,\ldots,\lambda_n)$ with $\Tr h = 0$, and define $a_R=\sum_{i\in R}\lambda_i$, $b_R=\sum_{i\in R}\lambda_i^2$ and $T=\sum_{i=1}^n\lambda_i^2$. 
Then, 
\begin{equation}
  \binom{n}{\eta}\int_{u\sim U(n)}\left(\Tr(uQ_Ru\ad h)\right)^2 \lvert\det u_{R,R}\rvert^2=\frac{\eta(n-\eta+1)}{n(n+1)(n+2)}T+\frac{n-2\eta}{n(n+1)(n+2)}b_R+\frac1{n(n+1)}a_R^2\,.
  \label{eq:bls}
\end{equation}
\end{restatable}
\begin{proof}
A convenient simplification arises from considering the decomposition $\mbc^n=R\oplus R^\perp$, and the corresponding $h=h_R\oplus h_C$. Similarly decomposing $uQ_Ru\ad$ into block form, $uQ_Ru\ad=\begin{pmatrix}
    X&*\\ *&Y
\end{pmatrix}$, we have $\Tr(uQ_Ru\ad h) = \Tr(X h_R) +\Tr(Y h_C) $ (with the off diagonal  blocks of $uQ_Ru\ad$ playing no role in the sequel), where $X=Q_RuQ_Ru\ad Q_R = u_{R,R}u_{R,R}\ad$ and $Y=(\id_n - Q_R)uQ_Ru\ad (\id_n - Q_R)$. Now, condition on $X$ having some fixed eigenvalues $x_1,\ldots,x_\eta$ and set $s_1=\sum_i x_i,$ and $s_2=\sum_i x_i^2$. By the invariance of the distribution whence $X$ came under the action of $K=U(\eta)\times U(n-\eta)$ we then have $X=V {\rm diag}\,(x_1,\ldots,x_\eta) V\ad$  for $V\sim U(\eta)$. Applying Eq.~\eqref{eq:fixedspec} then gives
\begin{equation}
\int_{V\sim U(\eta)} \left(\Tr [V {\rm diag}\,(x_1,\ldots,x_\eta) V\ad h_R]\right)^2 =   \frac{\eta s_1^2-s_2}{\eta (\eta ^2-1)}(\Tr h_R)^2+\frac{\eta s_2-s_1^2}{\eta (\eta ^2-1)}\Tr(h_R^2)
\end{equation}
Now, as we previously established, (seeing also Ref.~\cite{kieburg2016singular}), the eigenvalues of $X$ under the det-biased measure are distributed as the $\mu(x)$ of Lemma~\ref{lem:jactwo}, so that the result of that lemma yields
\begin{align}
\binom{n}{\eta}\int_{u\sim U(n)}\left(\Tr(X h_R)\right)^2 \lvert\det u_{R,R}\rvert^2&=\expect_{x\sim\mu} \left[\frac{\eta s_1^2-s_2}{\eta (\eta ^2-1)}(\Tr h_R)^2+\frac{\eta s_2-s_1^2}{\eta (\eta ^2-1)}\Tr(h_R^2)\right]\\
&=\frac{(\eta+1)(n\eta+n+\eta)}{n(n+1)(n+2)}(\Tr h_R)^2+\frac{(n-\eta)(\eta+1)}{n(n+1)(n+2)}\Tr(h_R^2)\label{eq:xterm}
\end{align}
Now, the eigenvalues of $Y$ are given by $y_i=1-x_i$ (with trailing zeroes for $i>\eta$), so that $\Tr Y=\eta-s_1$ and $\Tr Y^2 = \eta-2s_1 + s_2$. Applying Eq.~\eqref{eq:fixedspec} and Lemma~\ref{lem:jactwo} again then ensures
\begin{equation}
    \binom{n}{\eta}\int_{u\sim U(n)}\left(\Tr(Y h_C)\right)^2 \lvert\det u_{R,R}\rvert^2 =    \frac{\eta(n\eta+\eta-1)}{n(n+1)(n+2)}(\Tr h_C)^2+   \frac{\eta(n-\eta+1)}{n(n+1)(n+2)}\Tr(h_C^2)\label{eq:yterm}\,.
\end{equation}
We have also to compute a cross term; as (conditional on the eigenvalues $x$) the resulting freedoms are independent Haar random rotations on the eigenvectors we have 
\begin{align}
    \binom{n}{\eta}\int_{u\sim U(n)} \Tr(X h_R)\Tr(Y h_C)  \lvert\det u_{R,R}\rvert^2 &= \expect_{x\sim \mu} \left[\left(\expect_{V\sim U(\eta)} \Tr(V{\rm diag}\, (x)V\ad h_R)\right)\left(\expect_{V\sim U(n-\eta)} \Tr(V{\rm diag}\, (1-x)V\ad h_C)\right)\right] \\
     &= \expect_{x\sim \mu} \left[\left(\frac{s_1}{\eta}\Tr h_R\right)\left(\frac{\eta-s_1}{n-\eta}\Tr h_C\right)\right] \\
    &=  \frac{(\eta+1)(n\eta+\eta-1)}{n(n+1)(n+2)} \Tr h_R \Tr h_C\label{eq:cross}
\end{align}
(where ``$1-x$''  implicitly includes $n-2\eta$ padded zeros). 
Summing the three contributions Eqs.~\eqref{eq:xterm},~\eqref{eq:yterm} and~\eqref{eq:cross}, and using $\Tr h=0$, then swiftly yields the result of the lemma. 
\end{proof}

We are now in the favourable position of being able to derive a simple explicit expression for $\Gamma_h$ (Eq.~\eqref{eq:gammah}):

\begin{restatable}{lem}{lemgamhev}\label{lem:gamhev} 
 Let $h=h\ad$, $\Tr h=0$. Then
\begin{equation}
  \Gamma_h=A_{n,\eta}\Tr(h^2)\id_{\mch_\eta}
  +B_{n,\eta}d{U_\eta}(h^2)
  +C_n d{U_\eta}(h)^2\,,
\end{equation}
where $A_{n,\eta}=(n+1)\eta(n-\eta+1)/(n(n+2))$, $  B_{n,\eta}= (n+1)(n-2\eta)/(n(n+2))$, and $C_n=(n+1)/n$.
\end{restatable}
\begin{proof}
 Let us begin with the case where  $ h={\rm diag}\,(\lambda_1,\ldots,\lambda_n)$  is diagonal, so that  $h$ commutes with the diagonal torus $T^n=\{{\rm diag}\,(e^{i\theta_1},\ldots,e^{i\theta_n})\}_{\theta_1,\ldots \theta_n\in \mbr}$. Now, the readily verified equivariance $\Gamma_{vhv\ad}=U_\eta(v)\Gamma_hU_\eta(v)\ad$ of the second-moment operator applied to a $v\in T^n$ immediately implies that  $\Gamma_h$ commutes with $U_\eta(v)$ for every diagonal unitary $v$. On the other hand, the occupation basis is multiplicity-free for this torus; indeed, $U_\eta(v)\ket S=e^{i\sum_{i\in S}\theta_i}\ket S$, so that different choices of  $S$ specify different characters. We conclude that  any operator commuting with all $U_\eta(v)$, $v\in T^n$, is diagonal in the occupation basis. In particular  $\Gamma_h$ is diagonal in the occupation basis, and it suffices to compute $\bra S\Gamma_h\ket S$ for all $S\in\mcz_\eta$. We calculate:
\begin{align}
  \bra S\Gamma_h\ket S&=\int_{u\sim U(n)}\sum_{Z\in\mcz_\eta}\widehat{o}_h(u,Z)^2\bra S U_\eta(u)\ad\Pi_ZU_\eta(u)\ket S\\
  &=(n+1)^2\int_{u\sim U(n)}\sum_{Z\in\mcz_\eta} \Tr[h u\ad Q_Z u] ^2\bra S U_\eta(u)\ad\Pi_ZU_\eta(u)\ket S\\
  &=(n+1)^2\binom{n}{\eta}\int_{u\sim U(n)}  \Tr[h u\ad Q_S u]^2 \bra S U_\eta(u)\ad\Pi_SU_\eta(u)\ket S\\
  &=(n+1)^2\binom{n}{\eta}\int_{u\sim U(n)}  \Tr[h u\ad Q_S u]^2 \lvert\det u_{S,S}\rvert^2
\end{align}
where we have used Eq.~\eqref{eq:hest} and the observation that the transitivity of the action of $U_\eta(u)$ on the occupation basis renders the sum independent of $Z$, so that we may as well set $Z=S$ and multiply by the number $\binom{n}{\eta}$ of elements of the sum. Lemma~\ref{lem:jacthree} now gives
\begin{equation}
  \bra S\Gamma_h\ket S=A_{n,\eta}\sum_i\lambda_i^2+B_{n,\eta}\sum_{i\in S}\lambda_i^2+C_n\left(\sum_{i\in S}\lambda_i\right)^2\,,\label{eq:geigenvals}
\end{equation}
establishing the lemma for diagonal $h$. In general, let $h$ be diagonalised as $h=v\Lambda v\ad$. Then equivariance and the lemma applied to the diagonal operator $\Lambda$ give
\begin{align}
    \Gamma_h = \Gamma_{v\Lambda v\ad} = U_\eta(v)\Gamma_\Lambda U_\eta(v\ad) &= U_\eta(v)\left(A_{n,\eta}\Tr(\Lambda^2)\id_{\mch_\eta}+B_{n,\eta}d{U_\eta}(\Lambda^2)+C_n d{U_\eta}(\Lambda)^2\right) U_\eta(v\ad) \\
    &= \left(A_{n,\eta}\Tr(\Lambda^2)\id_{\mch_\eta}+B_{n,\eta}d{U_\eta}(v\Lambda^2v\ad)+C_n d{U_\eta}(v\Lambda v\ad)^2\right)\\ 
    &= \left(A_{n,\eta}\Tr(h^2)\id_{\mch_\eta}+B_{n,\eta}d{U_\eta}(h^2)+C_n d{U_\eta}(h)^2\right) \,,
\end{align}
concluding the proof. 
\end{proof}
\noindent
We are now ready to prove:
\thmham*
\begin{proof}
    As usual~\cite{huang2020predicting}, we have $\Var[\hat h] = \mbe[\hat{h}^2]-\mbe[\hat{h}]^2 = \mbe[\hat{h}_0^2]-\mbe[\hat{h}_0]^2\le \mbe[\hat{h}_0^2]\le \|\Gamma_{h_0}\|_\infty$, so that it suffices to bound the spectral norm of  the second moment operator (Eq.~\eqref{eq:gammah}) corresponding to    the traceless $h_0=h-\Tr[h]\id_n/n$.  Now, if the eigenvalues of $h_0$ are $\lm_1,\ldots \lm_n$, Lemma~\ref{lem:gamhev} immediately gives
    \begin{equation}
    \|\Gamma_{h_0}\|_\infty=
\max_{S\in\mcz_\eta}\left[A_{n,\eta}\sum_i\lambda_i^2+B_{n,\eta}\sum_{i\in S}\lambda_i^2+C_n\left(\sum_{i\in S}\lambda_i\right)^2
  \right]\, .\label{eq:gaminf}
    \end{equation}
So, consider some $S\in\mcz_\eta$. The  $S$-dependent part of \eqref{eq:gaminf} is the quadratic form $  q_S(\lambda)
  =B_{n,\eta}\sum_{i\in S}\lambda_i^2+C_n\left(\sum_{i\in S}\lambda_i\right)^2$, restricted by the tracelessness of $h_0$ to the hyperplane $\boldsymbol{1}^\perp=  \left\{\lambda\in\mbr^n: \sum_i\lambda_i=0\right\}$. We will produce a bound of the form $  q_S(\lambda)\le\kappa_{n,\eta}\norm{\lambda}_2^2$, for a $\kappa_{n,\eta}$ independent of $\lm$ (but dependent on $n $ and $\eta$). In order to analyse $q_S(\lambda)$ it is helpful to introduce the decomposition $  \boldsymbol{1}^\perp=W_S\oplus W_{S^c}\oplus L_S$, where
  \begin{align}
      W_S&=\Big\{\lambda:\sum_i\lambda_i=0,\ \ {\rm supp}\, \lm \subseteq S \Big\}\\
      W_{S^c}&=\Big\{\lambda:\sum_i\lambda_i=0,\ \ {\rm supp}\, \lm \subseteq [n]\setminus S \Big\}\\
  L_S&=
  \Big\{\lambda:\sum_i\lambda_i=0,\ \ \lambda_i=a\ \ \forall i\in S,\ \lambda_i=b\ \ \forall i\in [n]\setminus S \Big\}\,.
  \end{align}
As far as the analysis of $q_S(\lambda)$ is concerned, the virtue of this decomposition is that for the corresponding expansion  $\lm = u+v+w$ we have 
\begin{equation}
    q_S(u+v+w)=q_S(u)+q_S(v)+q_S(w)= B_{n,\eta}\|u\|_2^2 + B_{n,\eta}\eta a^2 + C_n \eta^2 a^2= B_{n,\eta}\|u\|_2^2 + \frac{(B_{n,\eta}+\eta C_n)(n-\eta)}{n} \|w\|_2^2\,,
\end{equation}
where we have denoted by $a$ the constant value taken by $w$ on $S$. So $q_S$ has eigenvalues $  0,
  B_{n,\eta},$ and $
  (B_{n,\eta}+C_n\eta)(n-\eta)/n$. One readily checks that $B_{n,\eta}\ge 0$ is always smaller than the third  eigenvalue, so that 
  \begin{align}
      \|\Gamma_{h_0}\|_\infty&=\max_{S\in\mcz_\eta}\left[A_{n,\eta}\sum_i\lambda_i^2+B_{n,\eta}\sum_{i\in S}\lambda_i^2+C_n\left(\sum_{i\in S}\lambda_i\right)^2\right]\\
&\le \left(A_{n,\eta} +\frac{(B_{n,\eta}+\eta C_n)(n-\eta)}{n}\right) \|\lm\|_2^2   \\
&= \left(  \frac{(n+1)\bigl(n(2\eta+1)-2\eta^2\bigr)}{n(n+2)}\right) \|\lm\|_2^2   \\
&\le ( 2\eta+1)  \|\lm\|_2^2  
  \end{align}
but $\|\lm\|_2^2   =\Tr[h_0^2]$, completing the proof. 
\end{proof}

\noindent
We now turn to the (very immediate) proofs of our last two results:
\thmhamscaling*
\begin{proof}
    With $h_0$ and $\ket\psi$ as defined we have $\braket{\psi|h_0}{\psi}=0$, so that $\Var [\hat{h}_0] = \mbe[\hat{h}_0^2]$. Both $\ketbra{S_+}$ and $\ketbra{S_-}$ are eigenvectors of $\Gamma_{h_0}$ with (from Eq.~\eqref{eq:geigenvals}) eigenvalue  $\eta(n+1)$, so that (using $\braket{S_+}{S_-}=0$) we have 
\begin{align}
    \Var [\hat{h}_0] = \frac{1}{2} \left(\bra{S_+} + \bra{S_-}\right) \Gamma_{h_0}\left(\ket{S_+} + \ket{S_-}\right)= \frac{1}{2} \left(\bra{S_+} + \bra{S_-}\right) \eta(n+1)\left(\ket{S_+} + \ket{S_-}\right)=\eta(n+1)\,.
\end{align}
As $\|h_0\|_2^2 = n$, we obtain the result. 
\end{proof}

\thmhamcomp*
\begin{proof}
Working with the estimator for a particle-preserving quadratic observable $h=\Tr[h]\id_n/n + h_0$ is considerably more straightforward than the estimator for Slater determinant overlaps. Indeed, for a sampled pair $(u,Z)$ we simply have (Eq.~\eqref{eq:hestfull}) 
\begin{equation}\label{eq:hest2}
    \widehat{o}_{h}(u,Z)=\eta\Tr[h]/n+(n+1)\Tr[ h_0u\ad Q_Z u] =\eta\Tr[h]/n+(n+1)\Tr[ h_0(u\ad)_{*,Z} u_{Z,*}] \,.
\end{equation}
Now, as $h\in\mbc^{n\times n},\,u_{Z,*}\in\mbc^{\eta\times n},$ and $(u\ad)_{*,Z}\in\mbc^{n\times \eta}$, we can perform the matrix multiplications in time $\mco(n^2 \eta )$; the trace term is of course but $\mco(n)$. 
\end{proof}
We remark that without this trick of expressing $u\ad Q_Z u$ in terms of a product of submatrices of $u$ and $u\ad$, the complexity of naively multiplying the matrices in Eq.~\eqref{eq:hest2} would scale as $\mco(n^3)$.

\end{document}